\documentclass[9pt,journal]{IEEEtran}
\ifCLASSINFOpdf
\else
\fi
\usepackage{amsmath,graphicx,amsthm,amssymb,mathrsfs}
\usepackage{epstopdf}
\usepackage{epsfig,array}
\usepackage{booktabs}
\usepackage{multirow}
\usepackage{mathrsfs}
\usepackage{bbm}
\usepackage{url}

\newtheorem{theorem}{Theorem}

\newtheorem{remark}{Remark}

\begin{document}
\title{ Efficient Sampling for Selecting Important Nodes in Random Network}

\author{ Haidong Li,~\IEEEmembership{}
        Xiaoyun Xu,~\IEEEmembership{} Yijie Peng,~\IEEEmembership{} and~ Chun-Hung Chen~\IEEEmembership{}
       % <-this % stops a space
\thanks{Haidong Li is with the Department
	of Industrial Engineering and Management, Peking University,
	Beijing, 100871 China e-mail:  haidong.li@pku.edu.cn. }

\thanks{Xiaoyun Xu is with the Department
	of Industrial Engineering and Management, Peking University,
	Beijing, 100871 China e-mail: xiaoyun.xu@pku.edu.cn.}

\thanks{Yijie Peng is with the Department
	of Industrial Engineering and Management, Peking University,
	Beijing, 100871 China e-mail:  pengyijie@pku.edu.cn.}% <-this % stops a space
% <-this % stops a space

\thanks{Chun-Hung~Chen is with the Department
	of System Engineering and Operations Research, George Mason University, Fairfax,
	VA, 22030 USA e-mail: cchen9@gmu.edu.}

}

\maketitle

\begin{abstract}
We consider the problem of selecting important nodes in a random network, where the nodes connect to each other randomly with certain transition probabilities. The node importance is characterized by the stationary probabilities of the corresponding nodes in a Markov chain defined over the network, as in Google's PageRank. Unlike deterministic network, the transition probabilities in random network are unknown but can be estimated by sampling. Under a Bayesian learning framework, we apply the first-order Taylor expansion and normal approximation to provide a computationally efficient posterior approximation of the stationary probabilities. In order to maximize the probability of correct selection, we propose a dynamic sampling procedure which uses not only posterior means and variances of certain interaction parameters between different nodes, but also the sensitivities of the stationary probabilities with respect to each interaction parameter. Numerical experiment results demonstrate the superiority of the proposed sampling procedure.
\end{abstract}


% \fi
\IEEEpeerreviewmaketitle

\section{Introduction}
We consider the problem of selecting the top $m$ important nodes from $n$ ($n>m$) nodes in a network, a central problem for many social and economical networks. In the World Wide Web, web pages and hyperlinks constitute a network, and Google's PageRank lists the most important web pages for each keyword \cite{Brin1998The}, \cite{Page1998The}; in sports events such as basketball and football, teams compete with each other and their win-loss relationship network helps determine which teams should be invited \cite{Govan2008Ranking}, \cite{Luke2007Ranking}; in social network like Twitter, the linking topology is used to rank members for popularity recommendation \cite{Weng2010TwitterRank}. Other examples include venture capitalists selection \cite{Bhat2012InvestorRank} and academic paper searching \cite{walker2007ranking},~\cite{jomsri2011citerank}.

A Markov chain is often used to describe the network. Specifically, each node (page/user/team) in the network is considered as a state of the Markov chain, and the nodes are linked randomly with certain transition probabilities. The node importance is ranked by the stationary probability of a Markov chain, which is the long-run proportion of visits to each state. A larger stationary probability indicates that the corresponding node is more important. Existing works consider a Markov chain with given transition probabilities, and focus on how to efficiently calculate stationary probabilities (see, e.g., \cite{Brin1998The}, \cite{Page1998The}, \cite{Langville2011Google}).
In practice, the transition probabilities are usually not apriori knowledge but estimated from the data. For instance, the hyperlinks between web pages on the Internet change dynamically; the relationship network among twitters is topic-specific; the competition results in sports are uncertain. Therefore, we focus on a random network with unknown transition probabilities.

In the random network, we sample the interactions between the nodes to estimate the transition probabilities as functions of certain interaction parameters, which is in turn to estimate the stationary probabilities. Since sampling could be  expensive, the total number of samples is usually limited. Moreover, the number of transition probabilities grows with the square of the number of nodes, so it would be practically infeasible to estimate all transition probabilities accurately for large-scale networks.  We consider a problem of maximizing the probability of correct selection (PCS) for selecting the top $m$ nodes subject to a fixed sampling budget. The estimation insecurities in different interaction parameters have heterogeneous effects on the PCS. The final ranking may be more sensitive to the perturbation in some interaction parameters. We aim to develop a dynamic sampling procedure to select the top $m$ nodes with a high statistical efficiency.

Our problem is closely related to the ranking and selection  (R\&S) problem well known in the field of simulation optimization \cite{Bechhofer1995Design}, which considers selecting the best or an optimal subset from a finite alternatives. There are the frequentist and Bayesian branches in  R\&S \cite{chen2011stochastic}. Sampling procedures in the frequentist branch allocate samples to guarantee a pre-specified PCS level ~(see, e.g., \cite{rinott1978two}, \cite{Koenig1985A}, \cite{kim2001fully}). The sampling procedures in the Bayesian branch aim to either maximize the PCS or minimize the expected opportunity cost subject to a given sampling budget (see, e.g., \cite{chen2000simulation}, \cite{Chick2001New}, \cite{peng2012efficient}).  Chen et al. \cite{chen2008efficient}, Zhang et al.~\cite{Zhang2012An}, and Gao and Chen~\cite{Gao2015A} study sampling procedures to maximize the PCS for selecting an optimal subset; Xiao and Lee \cite{Xiao2014Efficient} derive the convergence rate of the false subset-selection probability, and offer an allocation rule achieving an asymptotically optimal convergence rate; and Gao and Chen~\cite{gao2015efficient} develop a sampling procedure based on the expected opportunity cost. In R\&S, the alternatives are ranked by the expectations of their sample performance, which can be directly estimated by the sample average of each alternative,
whereas in our problem, the nodes are ranked by the stationary probabilities of the Markov chain, which are estimated indirectly from the interaction samples between different nodes.

In this research, a Bayesian estimation scheme is introduced to update the posterior belief on the unknown interaction parameters, and
an efficient posterior approximation of the stationary probability is derived by Taylor expansion and normal approximation. The asymptotic analysis of the normal approximation is provided. We propose a dynamic allocation scheme for Markov chain  (DAM) to efficiently select the top $m$ nodes, which myopically maximizes an approximation of the PCS and is proved to be consistent. The DAM uses not only posterior means and variances of certain interaction parameters between different nodes, but also the sensitivities of the stationary probabilities with respect to each interaction parameter.

The rest of this paper is organized as follows. In Section~\ref{sec_PD}, we formulate the problem. Section~\ref{sec_PDA} derives a posterior distribution approximation of the stationary probability. The DAM is proposed in Section~\ref{sec_DCP}, and numerical results are given in Section~\ref{sec_NR}. The last section concludes the paper and outlines future directions.

\section{Problem Formulation}\label{sec_PD}
\begin{figure*}[htbp]
  \begin{center}
    \includegraphics[width=5.0in]{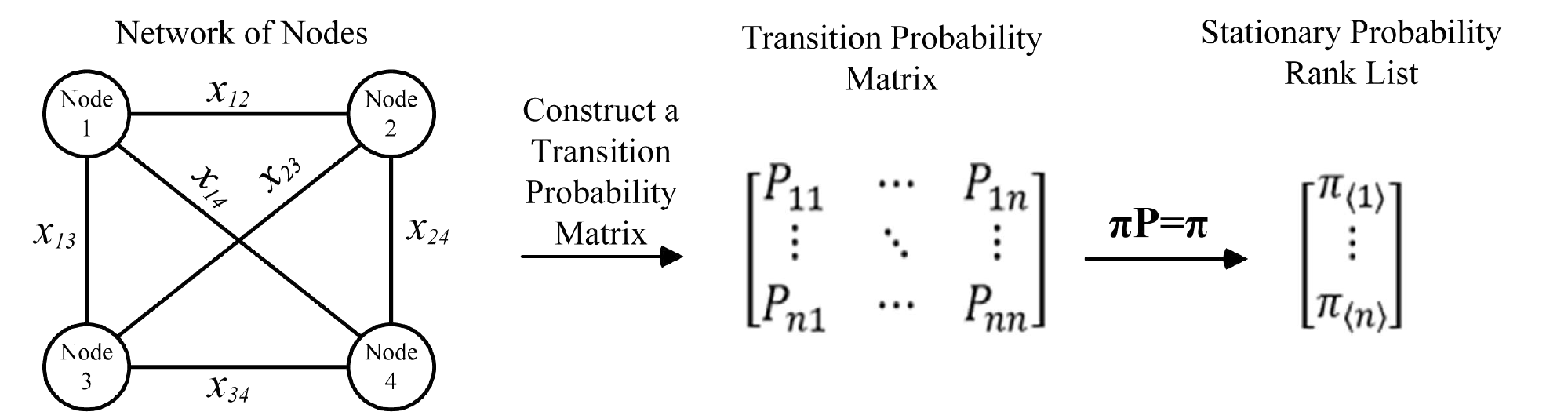}\\
    \caption{Illustration of node ranking process.}\label{pagerank}
  \end{center}
\end{figure*}
The objective of this study is to identify a subset of important nodes of a random network. The importance of the nodes is ranked by their stationary probabilities in a Markov chain. Specifically, $\pi_{i}$ denotes the stationary probability of node $i$, and our goal is to select the top $m$ nodes from $n$ nodes:
\begin{gather*}
\Theta_{m}\triangleq \left\{\langle 1\rangle,\ldots,\langle m\rangle\right\},
\end{gather*}
where notations $\langle i\rangle$, $i=1,\ldots,n$, are the ranking indices such that $\pi_{\langle 1\rangle}\geq\cdots\geq\pi_{\langle n \rangle}$.
The Markov chain of the network is constructed by the interaction strength between each pair of nodes. Given interaction parameter $x_{ij}$ describing the interaction strength between nodes $i$ and $j$, $1\leq i<j \leq n$, the transition probabilities $P_{ij}=P_{ij}(\mathscr{X})$,  $i,j=1, \ldots,n$, are   functions of a vector $\mathscr{X}\triangleq(x_{ij})_{1\leq i<j \leq n}$ with all interaction parameters as its elements. The functions appear in various forms for different applications \cite{Langville2012Who}. The transition matrix $\boldsymbol{\mathrm{P}}=[P_{ij}]_{n\times n}$ is a stochastic matrix satisfying irreducibility and aperiodicity, which guarantees the existence and uniqueness of the stationary distribution. The vector of stationary probabilities
$\boldsymbol{\mathrm{\pi}}\triangleq(\pi_{1},\ldots,\pi_{n})$ is the solution of the following equilibrium equation:
\begin{gather}
\boldsymbol{\mathrm{\pi P}}=\boldsymbol{\mathrm{\pi}},\ \text{and}\ \sum_{i=1}^{n} \pi_{i}=1,\ \pi_{i}>0.\label{eq_stationary}
\end{gather}
Notice that each stationary probability $\pi_{i}$ is also a function of $\mathscr{X}$. Figure~\ref{pagerank} summaries
the process of constructing a transition probability matrix
by the interaction parameters and ranking all nodes according to their stationary probabilities.

In a random network, the interactions between nodes are random. Specifically, let $X_{ij,t},t\in\mathbb{Z}^{+}$ be the $t$-th sample for the interactions between nodes $i$ and $j$, which is assumed to  follow an independent and identically distributed~(i.i.d.) Bernoulli distribution with unknown parameter $x_{ij}$, $1\leq i<j\leq n$. The Bernoulli assumption is natural in many practices that involve pairwise interaction. For instance, in the web page ranking, a binary variable takes $1$ for the visits from web page $j$ to $i$ and takes $0$ for the reverse direction; in the sport matches, the competition results of team $i$ and team $j$ take binary values $1$ (winning) and $0$ (losing). All interaction parameters $x_{ij}$, $1\leq i<j\leq n$, are assumed to be unknown but can be estimated by sampling. With the estimates of the interaction parameters, the transition probabilities and stationary probabilities can be in turn estimated.

Suppose the number of samples is fixed. The research problem of this work is to sequentially allocate each sample based on available information collected throughout previous sampling at each step to estimate the interaction strengthes between different pairs of nodes for efficiently selecting the top $m$ nodes. Given the information of $s$ allocated samples, the  selection is to pick the nodes with the top $m$ posterior estimates of the stationary probabilities, i.e.,
\begin{gather*}
\widehat{\Theta}_{m}^{(s)}\triangleq \left\{\langle1\rangle_{s},\ldots,\langle m\rangle_{s}\right\},
\end{gather*}
where $\langle i\rangle_{s}$, $i=1,\ldots,n$, are the ranking indices such that $$\pi_{\langle1\rangle_{s}}(\mathscr{X}^{(s)})\geq\cdots\geq\pi_{\langle n\rangle_{s}}(\mathscr{X}^{(s)}),$$ and $\mathscr{X}^{(s)}$ is a posterior estimate of $\mathscr{X}$ based on $s$ samples. We measure the statistical efficiency of a sampling procedure by the PCS defined as follows:
\begin{gather*} \mathrm{Pr}\left(\widehat{\Theta}_{m}^{(s)}=\Theta_{m}\right).
\end{gather*}
%Furthermore, another primary problem is how to allocate samples to efficiently improve PCS.

\section{Posterior of Stationary Probabilities}\label{sec_PDA}
We introduce a Bayesian framework to obtain posterior estimates of the stationary probabilities. From the Bayes rule, the posterior distribution of $\pi_{k}$ is
\begin{gather*}
F(d \pi_{k}|\mathcal{E}_{t})\triangleq \frac{L(\mathcal{E}_{t};\pi_{k})F(d \pi_{k}|\zeta_{0})}{\int L(\mathcal{E}_{t};\pi_{k})F(d \pi_{k}|\zeta_{0})},
\end{gather*}
where $\zeta_{0}$ is parameter in prior distribution, $\mathcal{E}_{t}$ is the information collected throughout the $t$-th sample, and $L(\cdot;\pi_{k})$ is the likelihood function of observed samples. In our problem, the likelihood function $L(\cdot;\pi_{k})$ does not have a closed form, so it is computationally  challenging to calculate the posterior distribution of the stationary distribution. To address the problem, we propose an efficient approximation for the posterior distributions of stationary probabilities by using the first-order Taylor expansion:
\begin{equation}\label{Taylor}
\begin{aligned}
\pi_{k}(\mathscr{X})\thickapprox&\pi_{k}(\mathscr{X}^{(t)})\\
&+\sum_{1\leq i<j\leq n}\left[\frac{\partial \pi_{k}(\mathscr{X})}{\partial x_{ij}}\bigg|_{\mathscr{X}=\mathscr{X}^{(t)}}\left(x_{ij}-x_{ij}^{(t)}\right)\right].
\end{aligned}
\end{equation}
Section~\ref{ssec_VE} will provide details in calculating $x_{ij}^{(t)}$, $\pi_{k}(\mathscr{X}^{(t)})$, and $\frac{\partial \pi_{k}(\mathscr{X})}{\partial x_{ij}}\big|_{\mathscr{X}=\mathscr{X}^{(t)}}$, and in Section~\ref{ssec_NA}, we will provide a normal approximation for the posterior of $\pi_{k}$.

\subsection{Posterior of Interaction Parameters}\label{ssec_VE}
Suppose the prior distribution of $x_{ij}$ follows a non-informative prior $U[0,1]$.  By conjugacy \cite{Gelman2014Bayesian}, the posterior distribution of $x_{ij}$  is a Beta distribution $Beta(\alpha_{ij}^{(t)},\beta_{ij}^{(t)})$ with the density given by
$$f\left(x,\alpha_{ij}^{(t)},\beta_{ij}^{(t)}\right)=\frac{\Gamma(\alpha_{ij}^{(t)}+\beta_{ij}^{(t)})}{\Gamma(\alpha_{ij}^{(t)})\Gamma(\beta_{ij}^{(t)})}x^{\alpha_{ij}^{(t)}-1}(1-x)^{\beta_{ij}^{(t)}-1},$$
where $$\alpha_{ij}^{(t)}\triangleq1+\sum_{\ell=1}^{t_{ij}}X_{ij,\ell},\quad \beta_{ij}^{(t)}\triangleq1+\sum_{\ell=1}^{t_{ij}}(1-X_{ij,\ell}),$$ and $t_{ij}$ is the number of samples allocated to estimate $x_{ij}$ after allocating $t$ samples in total, the posterior mean is
\begin{gather*}
x_{ij}^{(t)}\triangleq\alpha_{ij}^{(t)}/(\alpha_{ij}^{(t)}+\beta_{ij}^{(t)}),\label{post_mean}
\end{gather*}
and the posterior variance is
\begin{gather*}
(\sigma_{ij}^{(t)})^{2}\triangleq\alpha_{ij}^{(t)}\beta_{ij}^{(t)}/\left[(\alpha_{ij}^{(t)}+\beta_{ij}^{(t)})^{2}(\alpha_{ij}^{(t)}+\beta_{ij}^{(t)}+1)\right]~.\label{post_var}
\end{gather*}

Let $\mathscr{X}^{(t)}\triangleq(x_{ij}^{(t)})_{1\leq i<j \leq n}$ be a posterior estimate of $\mathscr{X}$.
Posterior estimates of the transition probability matrix $\boldsymbol{\mathrm{P}}$ and its derivative matrix are  $$\boldsymbol{\mathrm{P}}(\mathscr{X}^{(t)})=\left[P_{\ell k}(\mathscr{X}^{(t)})\right]_{n\times n}$$ and
$$\dfrac{\partial \boldsymbol{\mathrm{P}}(\mathscr{X})}{\partial x_{ij}}\bigg|_{\mathscr{X}^{(t)}}\triangleq\left[\frac{\partial P_{\ell k}(\mathscr{X})}{\partial x_{ij}}\bigg|_{\mathscr{X}=\mathscr{X}^{(t)}}\right]_{n\times n},$$
respectively.
Solving equilibrium equation~(\ref{eq_stationary}) by plugging in  $\boldsymbol{\mathrm{P}}(\mathscr{X}^{(t)})$ yields a posterior estimate of the vector of the stationary probabilities:
$$\boldsymbol{\mathrm{\pi}}(\mathscr{X}^{(t)})=\left({\pi}_{1}(\mathscr{X}^{(t)}),\ldots,{\pi}_{n}(\mathscr{X}^{(t)})\right).$$
By taking derivatives on both sides of the equilibrium equation $\pi_{k}(\mathscr{X})=\sum_{\ell=1}^{n}\pi_{\ell}(\mathscr{X})P_{\ell k}(\mathscr{X})$ with respect to $x_{ij}$, we have
\begin{align*}
	\frac{\partial\pi_{k}(\mathscr{X})}{\partial x_{ij}}=&\frac{\partial\sum_{\ell=1}^{n}\pi_{\ell}(\mathscr{X})P_{\ell k}(\mathscr{X})}{\partial x_{ij}}\\
	 =&\sum_{\ell=1}^{n}\left(\frac{\partial\pi_{\ell}(\mathscr{X})}{\partial x_{ij}}P_{\ell k}(\mathscr{X})+\pi_{\ell}(\mathscr{X})\frac{\partial P_{\ell k}(\mathscr{X})}{\partial x_{ij}}\right)\\
	 =&\sum_{\ell=1}^{n}\frac{\partial\pi_{\ell}(\mathscr{X})}{\partial x_{ij}}P_{\ell k}(\mathscr{X})+\sum_{\ell=1}^{n}\pi_{\ell}(\mathscr{X})\frac{\partial P_{\ell k}(\mathscr{X})}{\partial x_{ij}}.
\end{align*}
The derivative of the stationary distribution vector denoted by
$$\frac{\partial\boldsymbol{\mathrm{\pi}}(\mathscr{X})}{\partial x_{ij}}\triangleq\bigg(\frac{\partial\pi_{1}(\mathscr{X})}{\partial x_{ij}},\ldots,\frac{\partial\pi_{n}(\mathscr{X})}{\partial x_{ij}}\bigg)$$
is a solution of the following set of equations:
\begin{gather}\label{derivative_pi}
\left\{
\begin{array}{rcl}
\dfrac{\partial\boldsymbol{\mathrm{\pi}}(\mathscr{X})}{\partial x_{ij}}\left[\boldsymbol{\mathrm{I}}-\boldsymbol{\mathrm{P}}(\mathscr{X})\right]&=&\boldsymbol{\mathrm{{\pi}}}(\mathscr{X})\dfrac{\partial \boldsymbol{\mathrm{P}}(\mathscr{X})}{\partial x_{ij}},\\
\displaystyle{\sum\limits_{k=1}^{n}}\dfrac{\partial\pi_{k}(\mathscr{X})}{\partial x_{ij}}&=&0~.
\end{array}
\right.
\end{gather}
By plugging in $\mathscr{X}^{(t)}$, we have a posterior estimate of the derivative vector of the stationary probabilities:
$$\dfrac{\partial\boldsymbol{\mathrm{\pi}}(\mathscr{X})}{\partial x_{ij}}\bigg|_{\mathscr{X}^{(t)}}\triangleq\left(\dfrac{\partial\pi_{1}(\mathscr{X})}{\partial x_{ij}},\ldots,\dfrac{\partial\pi_{n}(\mathscr{X})}{\partial x_{ij}}\right)\bigg|_{\mathscr{X}=\mathscr{X}^{(t)}}~.$$

\begin{remark}
	In order to calculate  $\boldsymbol{\mathrm{{\pi}}}(\mathscr{X}^{(t)})$ and  $\dfrac{\partial\boldsymbol{\mathrm{\pi}}(\mathscr{X})}{\partial x_{ij}}\bigg|_{\mathscr{X}^{(t)}}$, we need to solve the linear equations involving the transition matrix of a Markov chain. Numerous efficient methods of industrial strength can be applied  \cite{Stewart1994Introduction}, \cite{barrett1994templates}. Google, for instance, has applied the power method to solve linear equations with a transition matrix of order $8.1$ billions \cite{moler2002world}.
\end{remark}

\subsection{Normal Approximation}\label{ssec_NA}
The posterior approximation of $\pi_{k}$ on the right hand side of (\ref{Taylor}) is a linear combination of $x_{ij}
\sim Beta(\alpha_{ij}^{(t)},\beta_{ij}^{(t)})$, $1\leq i<j\leq n$. Since Beta distributions are not closed in a linear combination, we use normal distribution $N\big(x_{ij}^{(t)}, {(\sigma_{ij}^{(t)})}^{2}\big)$ to approximate the posterior distribution of $x_{ij}$, which leads to a closed-form approximate posterior distribution of $\pi_{k}$.

Notice that $Beta(\alpha_{ij}^{(t)},\beta_{ij}^{(t)})$ and $N\big(x_{ij}^{(t)}, {(\sigma_{ij}^{(t)})}^{2}\big)$ share the same mean and variance. We further show that  $Beta(\alpha_{t},\beta_{t})$ converges in distribution to a normal distribution as $t\to +\infty$,  where $\alpha_{t}=\alpha_{0}t$ and $\beta_{t}=\beta_{0}t$ with $\alpha_{0},\beta_{0}>0$.
\begin{theorem}\label{thm_normal}
	As $t\rightarrow+\infty$,
	\begin{gather*}
	 \sqrt{t}\left[W_{t}-\frac{\alpha_{0}}{\alpha_{0}+\beta_{0}}\right]\xrightarrow{\ d\ }N\left(0,\frac{\alpha_{0}\beta_{0}}{(\alpha_{0}+\beta_{0})^{3}}\right),
	\end{gather*}
	where $W_{t}\sim Beta(\alpha_{t},\beta_{t})$ and $\xrightarrow{\ d\ }$ denotes convergence in distribution.
\end{theorem}
\begin{proof}
	A $Beta(\alpha_{t},\beta_{t})$ random variable can be represented as $Y_{t}/(Y_{t}+Z_{t})$, where $Y_{t}\sim Gamma(\alpha_{t},1)$, $Z_{t}\sim Gamma(\beta_{t},1)$ and $Y_{t}$ is independent of $Z_{t}$ \cite{Song2011Eighty}.
	Note that $Gamma(n,\lambda)$ can be represented as the sum of $n$ i.i.d. exponential random variables with parameter $\lambda$. By a Central Limit Theorem,
	\begin{gather*}
	 \sqrt{\alpha_{t}}\left[Y_{t}/\alpha_{t}-1\right]\xrightarrow{\ d\ }N(0,1)\ \text{as}\ t\rightarrow+\infty
	\end{gather*}
	and
	\begin{gather*}
	 \sqrt{\beta_{t}}\left[Z_{t}/\beta_{t}-1\right]\xrightarrow{\ d\ }N(0,1)\ \text{as}\ t\rightarrow+\infty.
	\end{gather*}
	Therefore,
	\begin{align*}
	\sqrt{t}&\left[\left(
	\begin{array}{c}
	Y_{t}/t\sqrt{\alpha_{0}} \\
	Z_{t}/t\sqrt{\beta_{0}} \\
	\end{array}
	\right)-\left(
	\begin{array}{c}
	\sqrt{\alpha_{0}} \\
	\sqrt{\beta_{0}} \\
	\end{array}
	\right)
	\right]\\&\qquad\xrightarrow{\ d\ }N\left(\left(
	\begin{array}{c}
	0 \\
	0 \\
	\end{array}
	\right),\left[
	\begin{array}{cc}
	1 & 0 \\
	0 & 1 \\
	\end{array}
	\right]
	\right)\ \text{as}\ t\rightarrow+\infty.
	\end{align*}
	With the multivariate delta method \cite{Cox2006Delta}, if there is a sequence of multivariate random variables $\boldsymbol{\mathrm{\theta_{n}}}$ satisfying $\sqrt{n}(\boldsymbol{\mathrm{\theta_{n}}}-\boldsymbol{\mathrm{\theta}})\xrightarrow{\ d\ }N(\boldsymbol{\mathrm{0}},\boldsymbol{\mathrm{\Sigma}})\ \text{as}\ n\rightarrow+\infty,$ where $\boldsymbol{\mathrm{\theta}}$ and $\boldsymbol{\mathrm{\Sigma}}$ are constant matrices, then for any continuously differentiable function $\boldsymbol{\mathrm{g}}(\cdot)$,
	\begin{gather*}
	 \sqrt{n}(\boldsymbol{\mathrm{g}}(\boldsymbol{\mathrm{\theta_{n}}})-\boldsymbol{\mathrm{g}}(\boldsymbol{\mathrm{\theta}}))\xrightarrow{\ d\ }N(\boldsymbol{\mathrm{0}},(\nabla\boldsymbol{\mathrm{g}}(\boldsymbol{\mathrm{\theta}}))^{T}\boldsymbol{\mathrm{\Sigma}}(\nabla\boldsymbol{\mathrm{g}}(\boldsymbol{\mathrm{\theta}})))\ \text{as} \ n\rightarrow+\infty.
	\end{gather*}
	Note that $W_{t}$ is a function of $Y_{t}/t\sqrt{\alpha_{0}}$ and $Z_{t}/t\sqrt{\beta_{0}}$, i.e., $$W_{t}=\dfrac{\sqrt{\alpha_{0}}(Y_{t}/t\sqrt{\alpha_{0}})}{\sqrt{\alpha_{0}}(Y_{t}/t\sqrt{\alpha_{0}})+\sqrt{\beta_{0}}(Z_{t}/t\sqrt{\beta_{0}})},$$ then
	\begin{gather*}
	 \sqrt{t}\left[W_{t}-\frac{\alpha_{0}}{\alpha_{0}+\beta_{0}}\right]\xrightarrow{\ d\ }N\left(0,\frac{\alpha_{0}\beta_{0}}{(\alpha_{0}+\beta_{0})^{3}}\right)\ \text{as}\ t\rightarrow+\infty,
	\end{gather*}
	which proves the conclusion.
\end{proof}

By the law of large number, $\alpha_{ij}^{(t)}, \beta_{ij}^{(t)}\sim O(t_{ij})$ almost surely (a.s.), as $t_{ij}\to+\infty$, where $A(x)=O(B(x))$  as $x\to\infty$ ($x\to0$) means that $|A(x)/B(x)|\to C>0$ as $x\to\infty$ ($x\to 0$). The asymptotic result in Theorem~\ref{thm_normal} justifies the asymptotic normality of $Beta(\alpha_{ij}^{(t)},\beta_{ij}^{(t)})$. Moreover, we show that the Kullback-Leibler (KL) divergence between $Beta(\alpha_{ij}^{(t)},\beta_{ij}^{(t)})$ and $N\big(x_{ij}^{(t)}, {(\sigma_{ij}^{(t)})}^{2}\big)$ goes to zero as $t_{ij}\to+\infty$. The KL divergence is a statistical (asymmetric) distance between two distributions \cite{Kullback1959Information}. Specifically, if $U$ and $V$ are probability measures over set $\Omega$, the KL divergence between $V$ and $U$ is defined by
\begin{gather*}
D_{KL}(U\|V)=\int_{\Omega}\log\frac{dU}{dV}dU.
\end{gather*}

\begin{theorem}\label{thm_converge}
	When $x_{ij}\neq 1/2$,
	$$D_{KL}(Beta\|Normal)=O\left(t_{ij}^{-1}\right)\quad a.s.\quad t_{ij}\to+\infty,$$
	and when $x_{ij}=1/2$,
	$$D_{KL}(Beta\|Normal)=o\left(t_{ij}^{-1}\right)\quad a.s.\quad t_{ij}\to+\infty,$$
	where $A(x)=o(B(x))$ means that $\underset{x\rightarrow+\infty}{\lim}A(x)/B(x)=0$.
\end{theorem}
\begin{proof}
	Let $$f\left(x,\alpha_{ij}^{(t)},\beta_{ij}^{(t)}\right)=\frac{\Gamma(\alpha_{ij}^{(t)}+\beta_{ij}^{(t)})}{\Gamma(\alpha_{ij}^{(t)})\Gamma(\beta_{ij}^{(t)})}x^{\alpha_{ij}^{(t)}-1}(1-x)^{\beta_{ij}^{(t)}-1},$$
	where $\Gamma(x)=\int_{0}^{+\infty}z^{x-1}e^{-z}dz$ is the Gamma function, and $$g\left(x,x_{ij}^{(t)},\sigma_{ij}^{(t)}\right)=\left(\sqrt{2\pi}\sigma_{ij}^{(t)}\right)^{-1}\exp\left(-\frac{(x-x_{ij}^{(t)})^{2}}{2{(\sigma_{ij}^{(t)})}^{2}}\right).$$ Then,
	\begin{align*}
		&D_{KL}(Beta\|Normal)\\
		 =&\int_{0}^{1}f\left(x,\alpha_{ij}^{(t)},\beta_{ij}^{(t)}\right)\log\frac{f\left(x,\alpha_{ij}^{(t)},\beta_{ij}^{(t)}\right)}{g\left(x,x_{ij}^{(t)},\sigma_{ij}^{(t)}\right)}dx\\
		=&\log e\times\int_{0}^{1}f\left(x,\alpha_{ij}^{(t)},\beta_{ij}^{(t)}\right)\ln\frac{f\left(x,\alpha_{ij}^{(t)},\beta_{ij}^{(t)}\right)}{g\left(x,x_{ij}^{(t)},\sigma_{ij}^{(t)}\right)}dx\\
		=&\log e\times\left(\int_{0}^{1}f\left(x,\alpha_{ij}^{(t)},\beta_{ij}^{(t)}\right)\ln f\left(x,\alpha_{ij}^{(t)},\beta_{ij}^{(t)}\right)dx\right.\\&\left.\qquad -\int_{0}^{1}f\left(x,\alpha_{ij}^{(t)},\beta_{ij}^{(t)}\right)\ln g\left(x,x_{ij}^{(t)},\sigma_{ij}^{(t)}\right)dx\right).
	\end{align*}
	From \cite{Cover1991Elements}, the entropy for the Beta distribution can be calculated by
	\begin{align*}
		 &\int_{0}^{1}f\left(x,\alpha_{ij}^{(t)},\beta_{ij}^{(t)}\right)\ln f\left(x,\alpha_{ij}^{(t)},\beta_{ij}^{(t)}\right)dx\\
		=&-\ln B(\alpha_{ij}^{(t)},\beta_{ij}^{(t)})+(\alpha_{ij}^{(t)}-1)(\psi(\alpha_{ij}^{(t)})-\psi(\alpha_{ij}^{(t)}+\beta_{ij}^{(t)}))\\
		 &\qquad+(\beta_{ij}^{(t)}-1)(\psi(\beta_{ij}^{(t)})-\psi(\alpha_{ij}^{(t)}+\beta_{ij}^{(t)})),
	\end{align*}
	where the digamma function $\psi(\cdot)$ is the first derivative of the log-gamma function, and $B(\alpha,\beta)=\Gamma(\alpha)\Gamma(\beta)/\Gamma(\alpha+\beta)$. By Stirling's formula, we have that as $\alpha_{ij}^{(t)},\beta_{ij}^{(t)}\rightarrow+\infty$,
	\begin{gather*}
		 \ln\Gamma(\alpha_{ij}^{(t)}+\beta_{ij}^{(t)})=\frac{1}{2}\ln2\pi+(\alpha_{ij}^{(t)}+\beta_{ij}^{(t)}-\frac{1}{2})\ln(\alpha_{ij}^{(t)}+\beta_{ij}^{(t)}) \\-(\alpha_{ij}^{(t)}+\beta_{ij}^{(t)})+\frac{1}{12}(\alpha_{ij}^{(t)}+\beta_{ij}^{(t)})^{-1}+o((\alpha_{ij}^{(t)}+\beta_{ij}^{(t)})^{-1});\\
		 \ln\Gamma(\alpha_{ij}^{(t)})=\frac{1}{2}\ln2\pi+(\alpha_{ij}^{(t)}-\frac{1}{2})\ln\alpha_{ij}^{(t)}-\alpha_{ij}^{(t)}+\frac{1}{12}(\alpha_{ij}^{(t)})^{-1}\\
     +o((\alpha_{ij}^{(t)})^{-1});\\
		 \ln\Gamma(\beta_{ij}^{(t)})=\frac{1}{2}\ln2\pi+(\beta_{ij}^{(t)}-\frac{1}{2})\ln\beta_{ij}^{(t)}-\beta_{ij}^{(t)}+\frac{1}{12}(\beta_{ij}^{(t)})^{-1}\\
     +o((\beta_{ij}^{(t)})^{-1})~.
	\end{gather*}
	With the results in \cite{Abramowitz1964Handbook}, as $\alpha_{ij}^{(t)},\beta_{ij}^{(t)}\rightarrow+\infty$, the digamma function has the following expansion:
	\begin{gather*}
	\psi(\alpha_{ij}^{(t)}+\beta_{ij}^{(t)})=\ln (\alpha_{ij}^{(t)}+\beta_{ij}^{(t)})-\frac{1}{2}(\alpha_{ij}^{(t)}+\beta_{ij}^{(t)})^{-1}\\-\frac{1}{12}(\alpha_{ij}^{(t)}+\beta_{ij}^{(t)})^{-2}+o((\alpha_{ij}^{(t)}+\beta_{ij}^{(t)})^{-2});\\
	\psi(\alpha_{ij}^{(t)})=\ln \alpha_{ij}^{(t)}-\frac{1}{2}(\alpha_{ij}^{(t)})^{-1}-\frac{1}{12}(\alpha_{ij}^{(t)})^{-2}+o((\alpha_{ij}^{(t)})^{-2});\\
	\psi(\beta_{ij}^{(t)})=\ln \beta_{ij}^{(t)}-\frac{1}{2}(\beta_{ij}^{(t)})^{-1}-\frac{1}{12}(\beta_{ij}^{(t)})^{-2}+o((\beta_{ij}^{(t)})^{-2}).
	\end{gather*}
	In addition,
	\begin{align*}
		 &\int_{0}^{1}f\left(x,\alpha_{ij}^{(t)},\beta_{ij}^{(t)}\right)\ln g\left(x,x_{ij}^{(t)},\sigma_{ij}^{(t)}\right)dx\\
		 =&-\int_{0}^{1}f\left(x,\alpha_{ij}^{(t)},\beta_{ij}^{(t)}\right)\left(\ln(\sqrt{2\pi}\sigma_{ij}^{(t)})+\frac{(x-x_{ij}^{(t)})^{2}}{2{(\sigma_{ij}^{2})}^{(t)}}\right)dx\\
		 =&-\ln\left(\sqrt{2\pi}\frac{(\alpha_{ij}^{(t)})^{\frac{1}{2}}(\beta_{ij}^{(t)})^{\frac{1}{2}}}{(\alpha_{ij}^{(t)}+\beta_{ij}^{(t)})(\alpha_{ij}^{(t)}+\beta_{ij}^{(t)}+1)^{\frac{1}{2}}}\right)-\frac{1}{2}.
	\end{align*}
	By the law of large numbers, we have
	$$\alpha_{ij}^{(t)},\beta_{ij}^{(t)}\rightarrow+\infty\ \ a.s.\ \text{when}\ t_{ij}\rightarrow+\infty$$
	and
	 $$\underset{t_{ij}\rightarrow+\infty}{\lim}\frac{\alpha_{ij}^{(t)}}{\alpha_{ij}^{(t)}+\beta_{ij}^{(t)}}=1-\underset{t_{ij}\rightarrow+\infty}{\lim}\frac{\beta_{ij}^{(t)}}{\alpha_{ij}^{(t)}+\beta_{ij}^{(t)}}= x_{ij}\ \ a.s..$$
	Therefore, as $t_{ij}\rightarrow+\infty$,
	\begin{align*}
		&D_{KL}(Beta\|Normal)\\
		=&\log e\times\left(\int_{0}^{1}f\left(x,\alpha_{ij}^{(t)},\beta_{ij}^{(t)}\right)\ln f\left(x,\alpha_{ij}^{(t)},\beta_{ij}^{(s)}\right)dx\right.\\&\left.\qquad-\int_{0}^{1}f\left(x,\alpha_{ij}^{(t)},\beta_{ij}^{(t)}\right)\ln g\left(x,x_{ij}^{(t)},\sigma_{ij}^{(t)}\right)dx\right)\\
		=&\log e\times\left(\frac{1}{2}\ln\left[1-(\alpha_{ij}^{(t)}+\beta_{ij}^{(t)}+1)^{-1}\right]+\frac{1}{3\alpha_{ij}^{(t)}}+\frac{1}{3\beta_{ij}^{(t)}}\right.\\
		 &\left.-\frac{5}{6(\alpha_{ij}^{(t)}+\beta_{ij}^{(t)})}+o\left(t_{ij}^{-1}\right)\right)\\
		=&\log e\times\left(\frac{1}{3x_{ij}(1-x_{ij})}-\frac{4}{3}\right)t_{ij}^{-1}+o\left(t_{ij}^{-1}\right)\quad a.s.,
	\end{align*}
	where the last equation holds due to the fact that $\ln(1-x^{-1})=-x^{-1}+o(x^{-1})$ as $x\rightarrow+\infty$.
	Notice that $$\dfrac{1}{3x_{ij}(1-x_{ij})}-\dfrac{4}{3}=0$$ if and only if $x_{ij}=1/2$. The conclusion follows immediately.
\end{proof}

\begin{remark}
	Notice that $D_{KL}(Beta\|Normal)$ converges at the fastest rate when $x_{ij}=1/2$. This could be explained by the fact that the normal distribution is a symmetrical distribution, and the Beta distribution is also a symmetrical distribution when $x_{ij}=1/2$. Numerical results show that $D_{KL}(Beta\|Normal)$ is close to zero even when $t_{ij}$ is not sufficiently large. For instance, when $t_{ij}=3$, $\alpha_{ij}^{(t)}=2$ and $\beta_{ij}^{(t)}=3$, the KL divergence between $Beta(2,3)$ and $N(2/5,1/25)$ is $0.0444$; when $t_{ij}=18$, $\alpha_{ij}^{(t)}=8$ and $\beta_{ij}^{(t)}=12$, the KL divergence between $Beta(8,12)$ and $N(2/5,2/175)$  is $0.0049$.
\end{remark}

The discussions above suggest that the statistical characteristics of $Beta(\alpha_{ij}^{(t)},\beta_{ij}^{(t)})$ and $N\big(x_{ij}^{(t)}, {(\sigma_{ij}^{(t)})}^{2}\big)$ are close when the allocated sample size is fairly large. Thus, we replace $x_{ij}\sim Beta(\alpha_{ij}^{(t)},\beta_{ij}^{(t)})$ with $\tilde{x}_{ij}\sim N\big(x_{ij}^{(t)}, {(\sigma_{ij}^{(t)})}^{2}\big)$ in (\ref{Taylor}). Then, the posterior approximation of $\pi_{k}$ in (\ref{Taylor}) is a linear combination of normal distributions, which follows the following normal distribution:
\begin{gather*}
\pi_{k}\sim N\left({\pi}_{k}^{(t)},(\tau_{k}^{(t)})^{2}\right),
\end{gather*}
where
$${\pi}_{k}^{(t)}\triangleq \pi_{k}(\mathscr{X}^{(t)}),$$
and
\begin{eqnarray*}
	(\tau_{k}^{(t)})^{2}\triangleq\sum_{1\leq i<j\leq n}\left[\left(\frac{\partial \pi_{k}(\mathscr{X})}{\partial x_{ij}}\bigg|_{\mathscr{X}=\mathscr{X}^{(t)}}\right)^{2}{(\sigma_{ij}^{(t)})}^{2}\right]~.
%	&=&\sum_{1\leq i<j\leq n}\left[\left(\frac{\partial \pi_{k}(\mathscr{X})}{\partial x_{ij}}\bigg|_{\mathscr{X}=\mathscr{X}^{(t)}}\right)^{2}\frac{\alpha_{ij}^{(t)}\beta_{ij}^{(t)}}{(\alpha_{ij}^{(t)}+\beta_{ij}^{(t)})^{2}(\alpha_{ij}^{(t)}+\beta_{ij}^{(t)}+1)}\right].
\end{eqnarray*}
When $n=2$, \begin{align*}(\tau_{1}^{(t)})^{2}&=\left(\frac{\partial \pi_{1}(\mathscr{X})}{\partial x_{1,2}}\bigg|_{\mathscr{X}=\mathscr{X}^{(t)}}\right)^{2}{(\sigma_{1,2}^{(t)})}^{2}\\
&=\left(\frac{\partial \pi_{2}(\mathscr{X})}{\partial x_{1,2}}\bigg|_{\mathscr{X}=\mathscr{X}^{(t)}}\right)^{2}{(\sigma_{1,2}^{(t)})}^{2}=(\tau_{2}^{(t)})^{2}.\end{align*}
Note that variance $(\tau_{k}^{(t)})^{2}$ in the normal approximation for the posterior distribution of the stationary probability  is affected by both posterior variance $(\sigma_{ij}^{(t)})^{2}$ of $x_{ij}$ and  posterior estimate $\dfrac{\partial \pi_{k}(\mathscr{X})}{\partial x_{ij}}\bigg|_{\mathscr{X}=\mathscr{X}^{(t)}}$ for the derivative of $\pi_k$ with respect to $x_{ij}$. Obviously, increasing in the posterior variance of $x_{ij}$ will result in the increasing in the variance of $\pi_{k}$ in the posterior approximation. On the other hand, if stationary probability $\pi_{k}$ is insensitive to parameter insecurity in $x_{ij}$, i.e., $\partial \pi_{k}/\partial x_{ij}$ is  small, large variance in $x_{ij}$ may not lead to large variance in $\pi_{k}$.
\begin{table*}[htbp]
  \caption{The Influence of Estimation Errors in Tested Parameters.}\vspace*{2mm}\label{Tab1}
  \resizebox{\textwidth}{10mm}{
    \begin{tabular}{cccccccccccc}
      \toprule
      \multicolumn{3}{c}{}
      & \multicolumn{1}{c}{} & \multicolumn{2}{c}{True Value} & \multicolumn{1}{c}{} & \multicolumn{2}{c}{Estimation 1}
      & \multicolumn{1}{c}{} & \multicolumn{2}{c}{Estimation 2}\\
      \cline{1-12}
      \multicolumn{3}{c}{Parameter ($x_{1,2},x_{1,3},x_{2,3}$)}
      & \multicolumn{1}{c}{} & \multicolumn{2}{c}{(0.7,\ 0.35,\ 0.6)} & \multicolumn{1}{c}{} & \multicolumn{2}{c}{(0.7,\ 0.35+0.02,\ 0.6)}
      & \multicolumn{1}{c}{} & \multicolumn{2}{c}{(0.7,\ 0.35,\ 0.6+0.02)}\\
      \cline{1-12}
      \multicolumn{3}{c}{Stationary Probability ($\pi_{1},\pi_{2},\pi_{3}$)}
      & \multicolumn{1}{c}{} & \multicolumn{2}{c}{(0.3477,\ 0.2916,\ 0.3607)} & \multicolumn{1}{c}{} & \multicolumn{2}{c}{(0.3582,\ 0.2897,\ 0.3521)}
      & \multicolumn{1}{c}{} & \multicolumn{2}{c}{(0.3497,\ 0.2989,\ 0.3514)}\\
      \cline{1-12}
      \multicolumn{3}{c}{Order Statistics}
      & \multicolumn{1}{c}{} & \multicolumn{2}{c}{(3,\ 1,\ 2)} & \multicolumn{1}{c}{} & \multicolumn{2}{c}{(1,\ 3,\ 2)}
      & \multicolumn{1}{c}{} & \multicolumn{2}{c}{(3,\ 1,\ 2)}\\
      \bottomrule
  \end{tabular}}
\end{table*}
Thus,  posterior variance $(\sigma_{ij}^{(t)})^{2}$ of $x_{ij}$ is scaled by  posterior estimate $\dfrac{\partial \pi_{k}(\mathscr{X})}{\partial x_{ij}}\bigg|_{\mathscr{X}=\mathscr{X}^{(t)}}$ for the derivative of $\pi_k$ with respect to $x_{ij}$ in  $(\tau_{k}^{(t)})^{2}$.

\section{Dynamic Sampling for Markov Chain}\label{sec_DCP}
Given the posterior approximations of stationary probabilities, we try to derive an efficient dynamic sampling procedure based on an approximate PCS.
The PCS for selecting the top $m$ nodes can be expressed as
\begin{align*}
	 \text{PCS}=&\mathrm{Pr}\left(\widehat{\Theta}_{m}^{(s)}=\Theta_{m}\right)\\
	=&\mathrm{Pr}\left(\pi_{\langle i\rangle_{s}}>\pi_{\langle j\rangle_{s}}, ~ i\in \{1,\ldots,m\},~ j\in \{m+1,\ldots,n\}\right).
\end{align*}

Insecurity in estimating $x_{ij}$, $1\leq i<j\leq n$, could result in  insecurity in estimating $\pi_k$, $k=1,\ldots,n$, which in turn leads to a low PCS. A noticeable feature in estimating the stationary probabilities is that the marginal influence of $x_{ij}$'s estimation insecurity on the stationary probabilities and thus the PCS is heterogeneous. To be more specific, large perturbations in some interaction parameters may have little influence on the stationary probabilities, whereas small perturbations in other interaction parameters could cause significant changes in the rank of the stationary probabilities and thus greatly affect the PCS. Such heterogeneity can be demonstrated by the following simple example: consider a 3-node network which has three interaction parameters ($x_{1,2},x_{1,3},x_{2,3}$)  to be estimated. The second column of Table~\ref{Tab1} lists the true interaction parameters, the stationary probabilities, and the final ranking. Table~\ref{Tab1} shows that the perturbation in $x_{1,3}$ could cause a significant change in stationary probabilities (Estimation 1), which even leads to incorrect selection of the best node. On the other hand, the same perturbation in $x_{2,3}$ has little influence on correctly selecting the optimal node subset (Estimation 2). To enhance the PCS under a limited sample size, this heterogeneity needs to be taken into consideration in the design of the sampling scheme.

We aim to obtain a dynamic sampling policy $\boldsymbol{\mathrm{D}}_{s}$ to maximize the PCS:
\begin{gather}\label{eq_pf}
\underset{\boldsymbol{\mathrm{D}}_{s}}{\max}\ \ \mathrm{Pr}\left(\pi_{\langle i\rangle_{s}}>\pi_{\langle j\rangle_{s}}, ~i\in \{1,\ldots,m\}, ~j\in \{m+1,\ldots,n\}\right).
\end{gather}
The dynamic sampling policy $\boldsymbol{\mathrm{D}}_{s}$ is a sequence of maps $\boldsymbol{\mathrm{D}}_{s}(\cdot)=(D_{1}(\cdot),\ldots,D_{s}(\cdot))$. Based on information set $\mathcal{E}_{t-1}$, $1\leq t\leq s$,  $D_{t}(\mathcal{E}_{t-1})\in\left\{(i,j):~ 1\leq i<j\leq n\right\}$ allocates the $t$-th sample to estimate an interaction parameter $x_{ij}$, $1\leq i<j\leq n$. Similar to that in \cite{peng2014dynamic} and \cite{peng2017ranking}, the policy optimization problem such as~(\ref{eq_pf}) can be formulated as a stochastic control (dynamic programming) problem. The expected payoff for a sampling scheme~$\boldsymbol{\mathrm{D}}_{s}$ can be defined recursively by
\begin{align}\label{value_function}
&V_{s}(\mathcal{E}_{s};\boldsymbol{\mathrm{D}}_{s})
\triangleq\mathbb{E}\left[\mathbbm{1}\left\{\widehat{\Theta}_{m}^{(s)}=\Theta_{m}\right\}\Big|\mathcal{E}_{s}\right]\nonumber\\
=&\text{Pr}\left(\pi_{\langle i\rangle_{s}}>\pi_{\langle j\rangle_{s}}, i\in \{1,..,m\}, j\in \{m+1,..,n\}\Big|\mathcal{E}_{s}\right),
\end{align}
and for $0\leq t<s$,
\begin{align}
&V_{t}(\mathcal{E}_{t};\boldsymbol{\mathrm{D}}_{s})
\triangleq\mathbb{E}\left[V_{t+1}(\mathcal{E}_{t}\cup \{X_{ij,t+1}\};\boldsymbol{\mathrm{D}}_{s})\Big|\mathcal{E}_{t}\right]\Big|_{(i,j)=D_{t+1}(\mathcal{E}_{t})},\nonumber
\end{align}
where equation~(\ref{value_function}) is a posterior integrated PCS.
Then, the optimal sampling policy is well defined by
\begin{eqnarray*}\label{scp}
	\boldsymbol{\mathrm{D}}_{s}^{*}\ \triangleq\ \arg\underset{\boldsymbol{\mathrm{D}}_{s}}{\max}\  V_{0}(\zeta_{0};\boldsymbol{\mathrm{D}}_{s}),
\end{eqnarray*}
where $\zeta_{0}$ is prior information. It is important to note that the definition of decision variable in our study is different from the one in R\&S. For the R\&S problem, the decision is to choose an alternative $i$ in sampling, whereas our decision is to choose a pair of nodes $(i,j)$ in sampling.

In principle, the backward induction can be used to solve the stochastic control problem, but it suffers from curse-of-dimensionality (see \cite{peng2017ranking}). To address this issue, we adopt approximate dynamic programming (ADP) schemes which make dynamic decision based on a value function approximation (VFA) and keep learning the VFA with decisions moving forward \cite{powell2007approximate}.
%At step $t$, the value function for selecting the top-$m$ nodes based on the posterior information is
From Section~\ref{sec_PDA}, an approximation for the posterior distribution of $\pi_{k}$ conditioned on $\mathcal{E}_{t}$ is a normal distribution with mean ${\pi}_{k}^{(t)}$ and $(\tau_{k}^{2})^{(t)}$.
Therefore, the joint distribution of vector $$\big(\pi_{\langle1\rangle_{t}}-\pi_{\langle m+1\rangle_{t}},..,\pi_{\langle1\rangle_{t}}-\pi_{\langle n\rangle_{t}},..,\pi_{\langle m\rangle_{t}}-\pi_{\langle m+1\rangle_{t}},..,\pi_{\langle m\rangle_{t}}-\pi_{\langle n\rangle_{t}}\big)$$ follows a joint normal distribution with mean vector $$\big({\pi}_{\langle1\rangle_{t}}^{(t)}-{\pi}_{\langle m+1\rangle_{t}}^{(t)},..,{\pi}_{\langle1\rangle_{t}}^{(t)}-{\pi}_{\langle n\rangle_{t}}^{(t)},..,{\pi}_{\langle m\rangle_{t}}^{(t)}-{\pi}_{\langle m+1\rangle_{t}}^{(t)},..,{\pi}_{\langle m\rangle_{t}}^{(t)}-{\pi}_{\langle n\rangle_{t}}^{(t)}\big),$$ and covariance matrix $\Gamma'\Lambda\Gamma$, where $$\Lambda\triangleq diag((\sigma_{1,2}^{(t)})^{2},\ldots,(\sigma_{1,n}^{(t)})^{2}, (\sigma_{2,3}^{(t)})^{2},\ldots,(\sigma_{n-1,n}^{(t)})^{2}),$$ and $\Gamma\triangleq$\\\tiny \begin{align*}
&\begin{bmatrix}d_{1,2}^{(t)}(1,m+1) & \cdot\cdot & d_{1,2}^{(t)}(1,n) & \cdot\cdot & d_{1,2}^{(t)}(m,m+1) & \cdot\cdot & d_{1,2}^{(t)}(m,n)\\
\vdots & & \vdots & & \vdots & & \vdots\\
d_{1,n}^{(t)}(1,m+1) & \cdot\cdot & d_{1,n}^{(t)}(1,n) & \cdot\cdot & d_{1,n}^{(t)}(m,m+1) &\cdot\cdot & d_{1,n}^{(t)}(m,n)\\
d_{2,3}^{(t)}(1,m+1) & \cdot\cdot & d_{2,3}^{(t)}(1,n) & \cdot\cdot & d_{2,3}^{(t)}(m,m+1) & \cdot\cdot & d_{2,3}^{(t)}(m,n)\\
\vdots & & \vdots & & \vdots & & \vdots\\
d_{n-1,n}^{(t)}(1,m+1) &\cdot\cdot & d_{n-1,n}^{(t)}(1,n) &\cdot\cdot & d_{n-1,n}^{(t)}(m,m+1) & \cdot\cdot& d_{n-1,n}^{(t)}(m,n)\\
\end{bmatrix},\end{align*}\normalsize
where  matrix $\Lambda$ is a diagonal matrix whose dimensionality is the same as the number of interaction parameters, $\Gamma$ is $[n(n-1)/2]\times [m(n-m)]$ matrix, and for $i\in\{1,\ldots,m\}$, $j\in \{m+1,\ldots,n\}$, $1\leq r<q\leq n$,
$$d_{r,q}^{(t)}(i,j)\triangleq \frac{\partial \big(\pi_{{\langle i\rangle}_{t}}(\mathscr{X})-\pi_{{\langle j\rangle}_{t}}(\mathscr{X})\big)}{\partial x_{rq}}\bigg|_{\mathscr{X}=\mathscr{X}^{(t)}}~.$$
Elements in matrix $\Gamma$ reflect the posterior information on sensitivities of the differences in stationary probabilities with respect to $x_{ij}$.

To derive a dynamic sampling procedure with an analytical form, we use the same VFA technique developed in \cite{peng2017ranking}. At any step $t$, we treat the ($t+1$)-th step as the last step and try to maximize the expected value function by allocating the ($t+1$)-th sample to a pair $(i,j)$:
\begin{gather*}
\widetilde{V}_{t}(\mathcal{E}_{t};(i,j))\triangleq\mathbb{E}\left[\widetilde{V}_{t+1}(\mathcal{E}_{t}\cup \{X_{ij,t+1}\})\Big|\mathcal{E}_{t}\right],
\end{gather*}
where
\begin{align*}\label{value_function_2}
	&\widetilde{V}_{t+1}(\mathcal{E}_{t+1})\triangleq \\
	&\text{Pr}\big(\pi_{\langle i\rangle_{t+1}}>\pi_{\langle j\rangle_{t+1}}, ~ i\in \{1,\ldots,m\}, ~ j\in \{m+1,\ldots,n\}\big|\mathcal{E}_{t+1}\big).
\end{align*}
The posterior probability above is an integral of the multivariate standard normal density over a region encompassed by some hyperplanes. We approximate the posterior probability by an integral over a maximum  tangent inner ball in the integral region. See more details about this approximation in \cite{peng2017ranking}. By symmetry of the normal density, maximizing the integral over a maximum tangent inner ball is equivalent to maximizing the volume of the ball, which  has the following analytical formula:
\begin{align*}
	v(\mathcal{E}_{t+1})=
	 \underset{\substack{k\in\{1,\ldots,m\}\\\ell\in\{m+1,\ldots,n\}}}{\min}\eta_{k\ell}^{(t+1)}\bigg/\zeta_{k\ell}^{(t+1)},
\end{align*}
where 
$$\eta_{k\ell}^{(t)}\triangleq \left({\pi}_{\langle k\rangle_{t}}^{(t)}-{\pi}_{\langle \ell\rangle_{t}}^{(t)}+\epsilon\right)^2,$$
and
$$\zeta_{k\ell}^{(t)}\triangleq \sum\limits_{1\leq r<q\leq n}\left[d_{r,q}^{(t)}(k,\ell)~\sigma_{rq}^{(t)}\right]^{2}~.$$
Here, we introduce a small positive real number $\epsilon$, so that the volume of the ball $v(\mathcal{E}_{t+1})$ must be positive. In~\cite{peng2017ranking} where the samples follow normal distributions, the volume of the ball is positive a.s. However, in our study, since all samples follow Bernoulli distributions,  $({\pi}_{\langle i\rangle_{t}}^{(t)}-{\pi}_{\langle j\rangle_{t}}^{(t)})$ is a discrete random variable so that the event 
$$\left({\pi}_{\langle i\rangle_{t}}^{(t)}-{\pi}_{\langle j\rangle_{t+1}}^{(t)}\right)^{2}=0$$ happens with a positive probability. In other words, the hyperplanes encompassing the integral region could pass through the origin of space. In order to have a positive volume of the ball, we shift each hyperplane away from the origin by distance $\epsilon$, which is visualized in Figure~\ref{feature}. In implementation, $\epsilon$ is set as a small positive number, e.g, $0.0001$.

\begin{figure}[htbp]
	\begin{center}
		\includegraphics[width=3.5in]{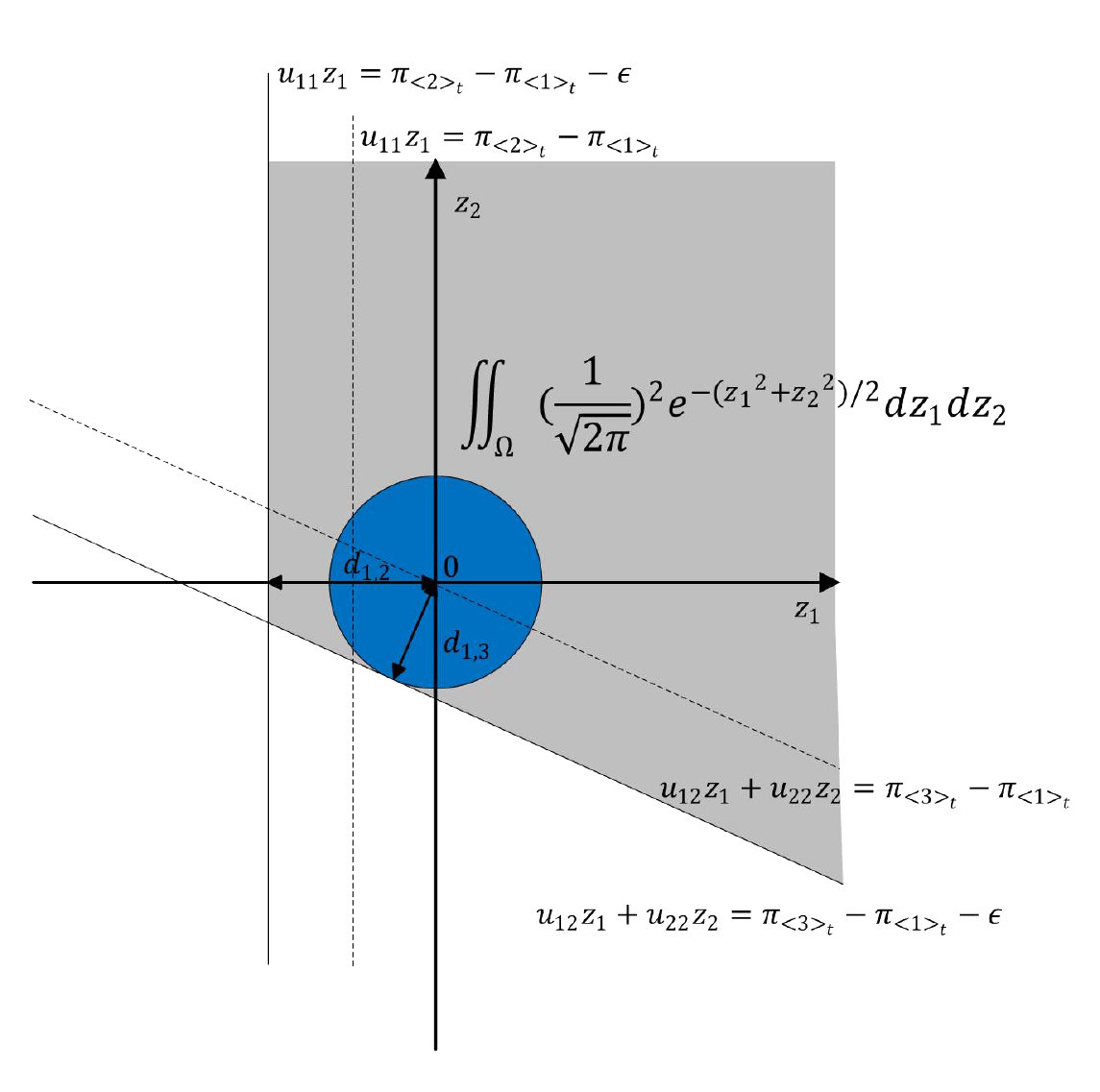}\\
		\caption{Illustration of the effect of $\epsilon$.}\label{feature}
	\end{center}
\end{figure}

By certainty equivalent approximation \cite{bertsekas1995dynamic},
\begin{gather*}
v\big(\mathcal{E}_{t}\cup \mathbb{E}[X_{ij,t+1}|\mathcal{E}_{t}]\big)
\approx \mathbb{E}\left[v\big(\mathcal{E}_{t}\cup \{X_{ij,t+1}\}\big)\Big|\mathcal{E}_{t}\right],
\end{gather*}
we have the following VFA:
\begin{eqnarray}
&&\widehat{V}_{t}(\mathcal{E}_{t};(i,j))\triangleq v(\mathcal{E}_{t}\cup \mathbb{E}[X_{ij,t+1}|\mathcal{E}_{t}])\nonumber\\
&=&\underset{\substack{k\in \{1,\ldots,m\}\\ \ell\in \{m+1,\ldots,n\}}}{\min}\frac{\eta_{k\ell}^{(t)}}{\sum\limits_{1\leq r<q\leq n}\left(d_{r,q}^{(t)}(k,\ell)\right)^{2}(\sigma_{rq}^{(t)})^{2}_{(i,j)}},\nonumber\\\label{AVFA}
\end{eqnarray}
where
\begin{gather*}
(\sigma_{rq}^{(t)})^{2}_{(i,j)}\triangleq \left\{
\begin{array}{rcl}
\frac{\alpha_{rq}^{(t)}\beta_{rq}^{(t)}}{(\alpha_{rq}^{(t)}+\beta_{rq}^{(t)})^{2}(\alpha_{rq}^{(t)}+\beta_{rq}^{(t)}+2)}, \text{when}\ (r,q)=(i,j);\\
\frac{\alpha_{rq}^{(t)}\beta_{rq}^{(t)}}{(\alpha_{rq}^{(t)}+\beta_{rq}^{(t)})^{2}(\alpha_{rq}^{(t)}+\beta_{rq}^{(t)}+1)}, \text{when}\ (r,q)\neq(i,j).
\end{array}
\right.
\end{gather*}

A dynamic allocation scheme for Markov chain (DAM) that optimizes the VFA  is given by
\begin{gather}\label{eq_policy}
D_{t+1}(\mathcal{E}_{t})=\arg \underset{1\leq i<j\leq n}{\max} \widehat{V}_{t}(\mathcal{E}_{t};(i,j)).
\end{gather}
The DAM uses the information on the posterior means of the stationary probabilities, which are calculated by the posterior means of the interaction parameters via equilibrium equation (\ref{eq_stationary}), the posterior variances of the interaction parameters, and the sensitivities of the stationary probabilities with respect to each interaction parameter.
Ignoring the small positive constant $\epsilon$, we note that 
$\eta_{k\ell}^{(t)}$ and $\zeta_{k\ell}^{(t)}$ are the squared mean and variance of the approximate posterior distribution of  the difference in the stationary probabilities, respectively. Therefore, equation (\ref{AVFA}) can be rewritten as $$\underset{\substack{k\in \{1,\ldots,m\}\\ \ell\in \{m+1,\ldots,n\}}}{\min}\ 1/c_{v}^{2}(k,\ell),$$ where $c_{v}(k,\ell)$ is the coefficient of variation (CV, or sometimes called noise-signal ratio) of the posterior approximation for $\left(\pi_{{\langle k\rangle}_{t}}-\pi_{{\langle \ell\rangle}_{t}}\right)$. The DAM  minimizes the maximum of $c_{v}(k,\ell)$'s, which is intuitively reasonable since large $c_{v}(k,\ell)$ implies high difficulty in comparing $\pi_{\langle k\rangle_{t}}$ and $\pi_{\langle \ell\rangle_{t}}$ from the posterior information. The DAM sequentially allocates each sample to estimate the interaction parameter to reduce the CV of the difference in each pair of the stationary probabilities. In particular, the DAM focuses on the pair most difficult in comparison among all possible pairs
in differentiating the top $m$ stationary probabilities from the rests based on the posterior information at each step.
The DAM is proved to be consistent in the following theorem.
\begin{theorem}\label{thm_consistency}
If for $1\leq i<j\leq n$ and $t\in\mathbb{Z}^+$,
 $$\dfrac{\partial\boldsymbol{\mathrm{\pi}}(\mathscr{X})}{\partial x_{ij}}\bigg|_{\mathscr{X}^{(t)}}\neq \mathbf{0} \quad a.s.,$$ 
  then the DAM is consistent, i.e., $$\underset{s\rightarrow+\infty}{\lim}\widehat{\Theta}_{m}^{(s)}=\Theta_{m},\ a.s.$$
\end{theorem}
\begin{proof}
	We only need to prove that each $x_{ij}$ will be sampled infinitely often a.s. following DAM, and the consistency will follow by the law of large numbers. Suppose parameter $x_{ij}$ is only sampled finitely often and parameter $x_{rq}$ is sampled infinitely often. Therefore, there exists a finite number $N_{0}$ such that parameter $x_{ij}$ will stop receiving replications after the sampling number $t$ exceeds $N_{0}$. Thus we have $$\underset{t\rightarrow+\infty}{\lim}(\sigma_{ij}^{(t)})^2>0,\quad \underset{t\rightarrow+\infty}{\lim}(\sigma_{rq}^{(t)})^2=0~.$$  
	
	If there exists a pair ($k,\ell$), $k\in\{1,\ldots,m\}$, $\ell\in \{m+1,\ldots,n\}$ such that $$\underset{t\rightarrow+\infty}{\lim}\left[d_{i,j}^{(t)}(k,\ell)\right]^2>0,$$ then $$\underset{t\rightarrow+\infty}{\lim}v(\mathcal{E}_t)<+\infty~.$$ Consider the pair $$(k',\ell')\triangleq\underset{\substack{k\in \{1,\ldots,m\}\\ \ell\in \{m+1,\ldots,n\}}}{\arg \min}\underset{t\rightarrow+\infty}{\lim}\eta_{k\ell}^{(t)}/\zeta_{k\ell}^{(t)}~.$$ If $$\underset{t\rightarrow+\infty}{\lim}\left[d_{i,j}^{(t)}(k',\ell')\right]^2=0$$ holds for each parameter $x_{ij}$ which is only sampled finitely often, then $$\underset{t\rightarrow+\infty}{\lim}\eta_{k'\ell'}^{(t)}/\zeta_{k'\ell'}^{(t)}=+\infty,$$ which contradicts with $$\underset{t\rightarrow+\infty}{\lim}\frac{\eta_{k'\ell'}^{(t)}}{\zeta_{k'\ell'}^{(t)}}=\underset{\substack{k\in \{1,\ldots,m\}\\ \ell\in \{m+1,\ldots,n\}}}{\min}\underset{t\rightarrow+\infty}{\lim}\frac{\eta_{k\ell}^{(t)}}{\zeta_{k\ell}^{(t)}}=\underset{t\rightarrow+\infty}{\lim}v(\mathcal{E}_t)<+\infty~.$$ However, if $$\underset{t\rightarrow+\infty}{\lim}\left[d_{i,j}^{(t)}(k',\ell')\right]^2>0$$ holds for a certain parameter $x_{ij}$ which is only sampled finitely often, by noticing that $$\underset{t\rightarrow+\infty}{\lim}\left[(\sigma_{ij}^{(t)})^2-(\sigma_{ij}^{(t)})^2_{(i,j)}\right]>0,$$ and $$\underset{t\rightarrow+\infty}{\lim}\left[(\sigma_{rq}^{(t)})^2-(\sigma_{rq}^{(t)})^2_{(r,q)}\right]=0,$$
	we have 
	\begin{gather*}
	 \underset{t\rightarrow+\infty}{\lim}\left[\widehat{V}_{t}(\mathcal{E}_{t};(i,j))-v(\mathcal{E}_t)\right]>0~\ a.s.,
	\end{gather*}
	and
	\begin{gather*}
	 \underset{t\rightarrow+\infty}{\lim}\left[\widehat{V}_{t}(\mathcal{E}_{t};(r,q))-v(\mathcal{E}_t)\right]=0~\ a.s.,
	\end{gather*}
	which contradicts with the sampling rule in equation~(\ref{eq_policy}) that the parameter
	with the largest $\widehat{V}_{t}(\mathcal{E}_{t};(i,j))$ is sampled.
	
	 Therefore, $$\underset{t\rightarrow+\infty}{\lim}\left[d_{i,j}^{(t)}(k,\ell)\right]^2=0$$ holds for each pair ($k,\ell$), $k\in\{1,\ldots,m\}$, $\ell\in \{m+1,\ldots,n\}$, that is, for $ 1\leq k,\ell\leq n$ $$\underset{t\rightarrow+\infty}{\lim}\frac{\partial \pi_{k}(\mathscr{X})}{\partial x_{ij}}\bigg|_{\mathscr{X}=\mathscr{X}^{(t)}}=\underset{t\rightarrow+\infty}{\lim}\frac{\partial \pi_{\ell}(\mathscr{X})}{\partial x_{ij}}\bigg|_{\mathscr{X}=\mathscr{X}^{(t)}}~.$$ Since  $$\sum\limits_{k=1}^{n}\frac{\partial\pi_{k}(\mathscr{X})}{\partial x_{ij}}\bigg|_{\mathscr{X}=\mathscr{X}^{(t)}}=0,$$ 
	we have 
	 $$\underset{t\rightarrow+\infty}{\lim}\frac{\partial \pi_{k}(\mathscr{X})}{\partial x_{ij}}\bigg|_{\mathscr{X}=\mathscr{X}^{(t)}}=0,\quad  1\leq k\leq n,$$ which contradicts with $$\dfrac{\partial\boldsymbol{\mathrm{\pi}}(\mathscr{X})}{\partial x_{ij}}\bigg|_{\mathscr{X}^{(t)}}\neq \mathbf{0} \quad a.s.\quad 1\leq i<j\leq n,\quad  t\in\mathbb{Z}^+~.$$ Therefore, DAM must be consistent.
\end{proof}

\begin{remark}
The assumptions in Theorem~\ref{thm_consistency} can be checked for the Markov chain in Google's PageRank \cite{Langville2011Google}, where the transition probabilities are given by
\begin{align*}
&P_{ji}\triangleq x_{ij}/(n-1),\quad  1\leq i<j\leq n;\\
&P_{ij}\triangleq1/(n-1)-P_{ji},\quad i\neq j;\\
&P_{ii}\triangleq1-\sum_{j\neq i}P_{ij},\quad 1\leq i\leq n.
\end{align*}
This Markov chain is a random walk, which is irreducible and aperiodic. At each step, the current state (page) $j$ chooses another page with equal probability to interact, and if page $i$ is chosen, the next state will be $i$ with probability $x_{ij}$ or still stay in $j$ otherwise. The importance of each web page is described by the long-run proportion of time spent on each state, i.e., its stationary probability. 
For the transition matrix in PageRank,
\begin{gather*}
\frac{\partial \boldsymbol{\mathrm{P}}}{\partial x_{ij}}=\ \left\{
\begin{array}{rcl}
1/(n-1)&, &\text{for element}\ (i,i)\ \text{and}\ (j,i);\\
-1/(n-1)&, &\text{for element}\ (i,j)\ \text{and}\ (j,j);\\
0&, &\text{otherwise}.
\end{array}
\right.
\end{gather*}
From (\ref{derivative_pi}), 
 $$\dfrac{\partial\boldsymbol{\mathrm{\pi}}(\mathscr{X})}{\partial x_{ij}}\bigg|_{\mathscr{X}^{(t)}}=\mathbf{0}$$ is equivalent to $\pi_{i}^{(t)}+\pi_{j}^{(t)}=0$. By ergodicity of Markov chain, $$\pi_{i}^{(t)}+\pi_{j}^{(t)}>0\quad  a.s., \quad 1\leq i<j\leq n, \quad t\in\mathbb{Z}^+~.$$ Therefore,
  $$\dfrac{\partial\boldsymbol{\mathrm{\pi}}(\mathscr{X})}{\partial x_{ij}}\bigg|_{\mathscr{X}^{(t)}}\neq \mathbf{0}\quad  a.s., \quad 1\leq i<j\leq n, \quad t\in\mathbb{Z}^+~.$$
\end{remark}

\section{Numerical Results}\label{sec_NR}
In the numerical experiments, we test the performance of different sampling procedures for ranking node importance in the Markov chain of PageRank. 
The proposed DAM is compared with the  equal allocation (EA) and  an approximately optimal allocation (AOA) adapted from a sampling procedure for classic R\&S problem in \cite{peng2017ranking}. Specifically, EA equally allocates sampling budget to estimate each $x_{ij},\ 1\leq i<j\leq n$ (roughly $s/(n(n-1)/2)$ samples for each $x_{ij}$);  AOA  allocates samples according to the following rules:
\begin{gather*}
\widehat{A}_{t+1}(\mathcal{E}_{t})=\arg \underset{1\leq i<j\leq n}{\max} V_{t}(\mathcal{E}_{t};(i,j)),
\end{gather*}
where
\begin{eqnarray*}
	 V_{t}(\mathcal{E}_{t};(i,j))\triangleq\underset{\substack{k\in \{1,\ldots,m\}\\ \ell\in \{m+1,\ldots,n\}}}{\min}\frac{\eta_{k\ell}^{(t)}}{\sum\limits_{1\leq r<q\leq n}(\sigma_{rq}^{(t)})^{2}_{(i,j)}}.
\end{eqnarray*}
%and
%\begin{gather*}
%(\sigma_{rq}^{2})^{(t;(i,j))}=\left\{
%\begin{array}{rcl}
%\frac{\alpha_{rq}^{(t)}\beta_{rq}^{(t)}}{(\alpha_{rq}^{(t)}+\beta_{rq}^{(t)})^{2}(\alpha_{rq}^{(t)}+\beta_{rq}^{(t)}+2)}, \text{when}\ (r,q)=(i,j);\\
%\frac{\alpha_{rq}^{(t)}\beta_{rq}^{(t)}}{(\alpha_{rq}^{(t)}+\beta_{rq}^{(t)})^{2}(\alpha_{rq}^{(t)}+\beta_{rq}^{(t)}+1)}, \text{when}\ (r,q)\neq(i,j).
%\end{array}
%\right.
%\end{gather*}
Notice that the AOA only utilizes the information in the posterior means of the stationary probabilities and the posterior variances of the interaction parameters, but it does not consider the information in the sensitivities of the stationary probabilities with respect to each interaction parameter.
In all numerical examples, the statistical efficiency of the sampling procedures is measured by the PCS estimated by 10,000 independent experiments. The PCS is reported as a function of the sampling budget in each experiment. 

\subsection*{Example 1: selecting top-3 nodes in a 10-node network}
In this example, we aim to identify the top-3 nodes from a network of 10 nodes. Suppose the true value of each interaction parameter is $$x_{ij}=0.5+0.03\times(j-i),\quad 1\leq i<j\leq 10~.$$ As the assumption in Section~\ref{sec_PD}, the samples of the interaction parameter $x_{ij}$ are generated i.i.d. from a Bernoulli distribution with parameter $x_{ij}$. According to the definition of $x_{ij}$, node $i$ is visited more often in the interactions between nodes $i$ and $j$ when $x_{ij}>0.5$. It is straightforward to know that nodes $1$, $2$, $3$ are the top-$3$ nodes. 

In Figure~\ref{10-3}, we can see that AOA has a slight edge over EA, which could be attributed to the reason that EA utilizes no sample information while AOA utilizes the information in the posterior means and variances, and DAM performs significantly better than the other two sampling procedures. In order to attain PCS = 80\%, DAM needs less than 1500 samples, whereas EA and AOA require more than 2000 samples. That is to say DAM reduces the sampling budget by more than 25\%. 
The performance enhancement of DAM could be attributed to the utilization of not only the information in posterior means and variances but also the sensitivity information (${\partial \pi_{k}}/{\partial x_{ij}}$). The numerical result shows that in this example,  the sensitivity information plays a dominant role in enhancing the sampling efficiency.

\begin{figure}[htbp]
	\begin{center}
		\includegraphics[width=3.5in]{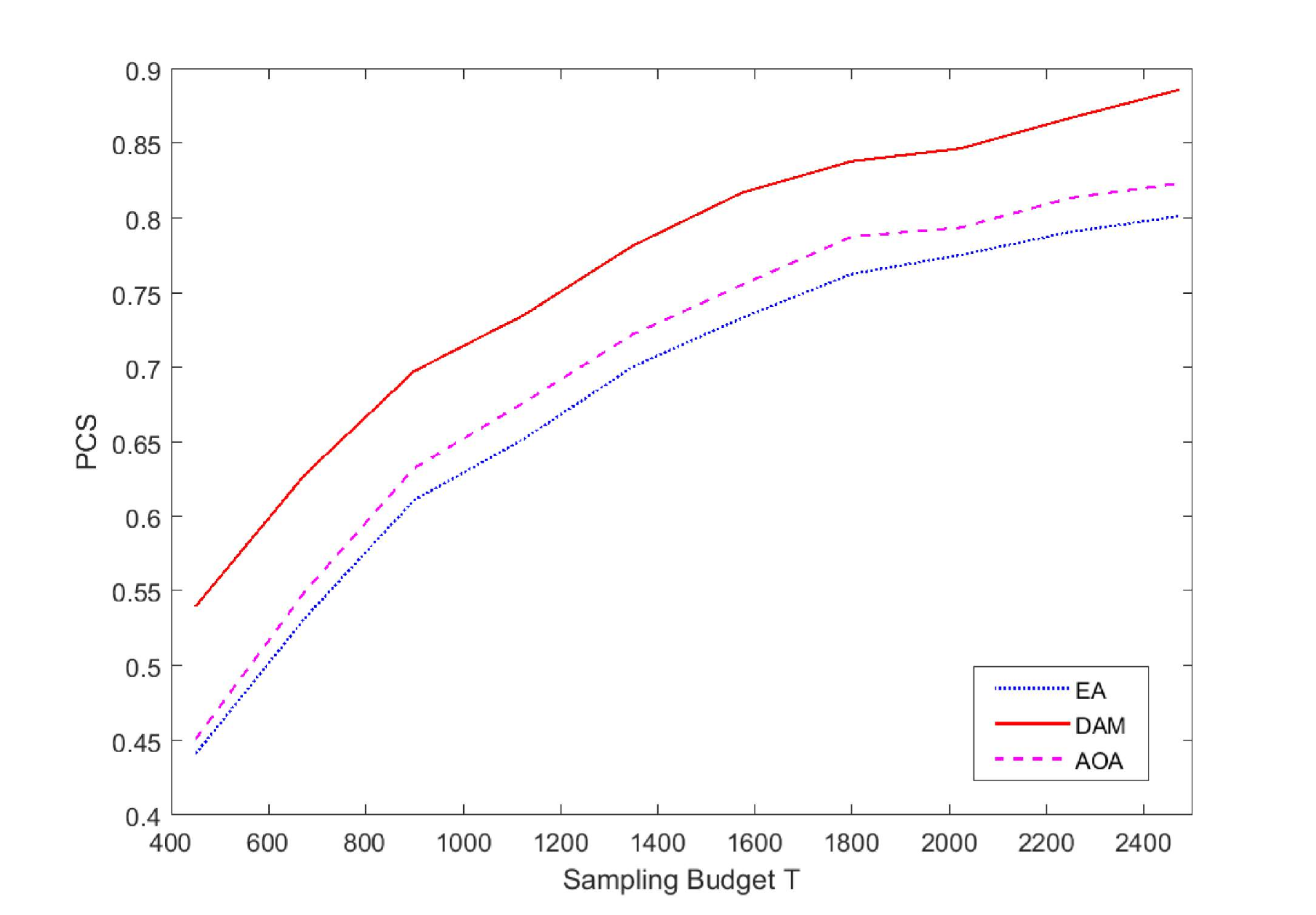}\\
		\caption{PCS of the three sampling procedures in Example 1.}\label{10-3}
	\end{center}
\end{figure}

\begin{figure}[htbp]
	\begin{center}
		\includegraphics[width=3.5in]{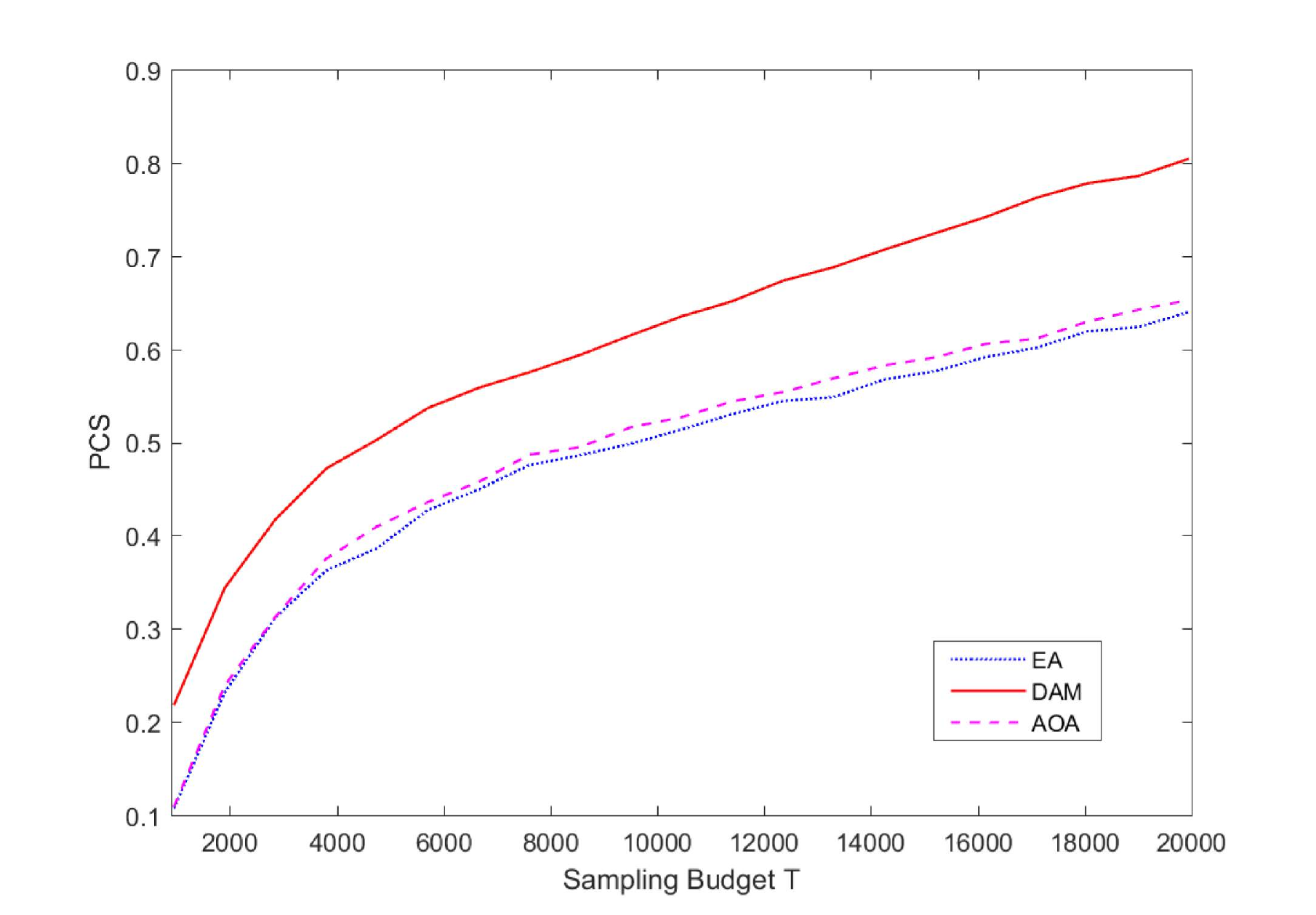}\\
		\caption{PCS of the three sampling procedures in Example 2.}\label{20-5}
	\end{center}
\end{figure}

\subsection*{Example 2: selecting top-5 nodes in a 20-node network}
In this example, we test the performance of the proposed DAM in a larger scale network with 20 nodes. The true value of each interaction parameter $x_{ij}$, $1\leq i<j\leq 20$, is drawn from a uniform prior distribution $U[0,1]$. Our objective is to identify the optimal subset of nodes with the top-5 largest stationary probabilities. Figure~\ref{20-5} illustrates the performance of the three sampling procedures. Similar to Example 1, DAM remains as the most efficient sampling procedure among the three, and AOA is slightly better than EA. However, it can be noticed that the advantage of DAM is more significant when the network size becomes larger. In order to attain PCS = 60\%, the number of samples consumed by DAM is less than 9000, while both EA and AOA require more than 15000 samples.  That is to say DAM reduces the sampling budget by more than 40\%. 

\subsection*{Example 3: selecting top-15 nodes in a 105-node website network}
In this example, we test the robustness for the performance  of DAM in a real data set from the Sogou Labs, a major web searching engine company in China  (\url{http://www.sogou.com/labs/resource/t-link.php}). The data set includes a mapping table from URL to document ID and a list of hyperlink relationship of the documents. Our objective is to select the top-15 websites from a $105$-node website network. The true value of each interaction parameter $x_{ij}$ is estimated from the data set. Figure~\ref{link} illustrates the interactions among the websites. For instance, the visits from website $j$ to $i$ occur 12 times, while the visits from website $i$ to $j$ only occur 7 times, so the true value of $x_{ij}$ is set as $12/(12+7)$. 
\begin{figure}[htbp]
	\begin{center}
		\includegraphics[width=3.5in]{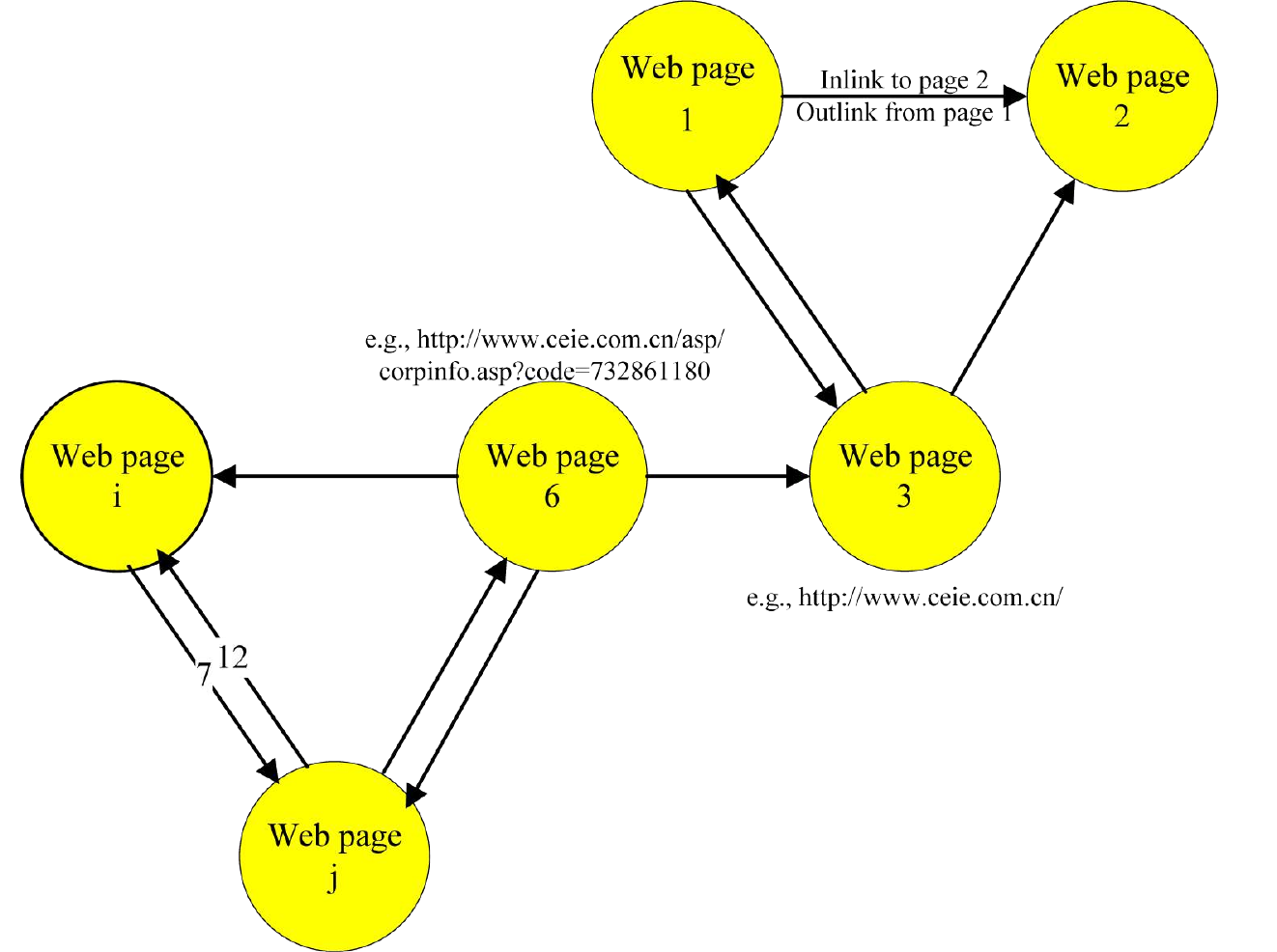}\\
		\caption{Interactions in Website Network from Sogou Labs.}\label{link}
	\end{center}
\end{figure}

\begin{figure}[htbp]
	\begin{center}
		\includegraphics[width=3.5in]{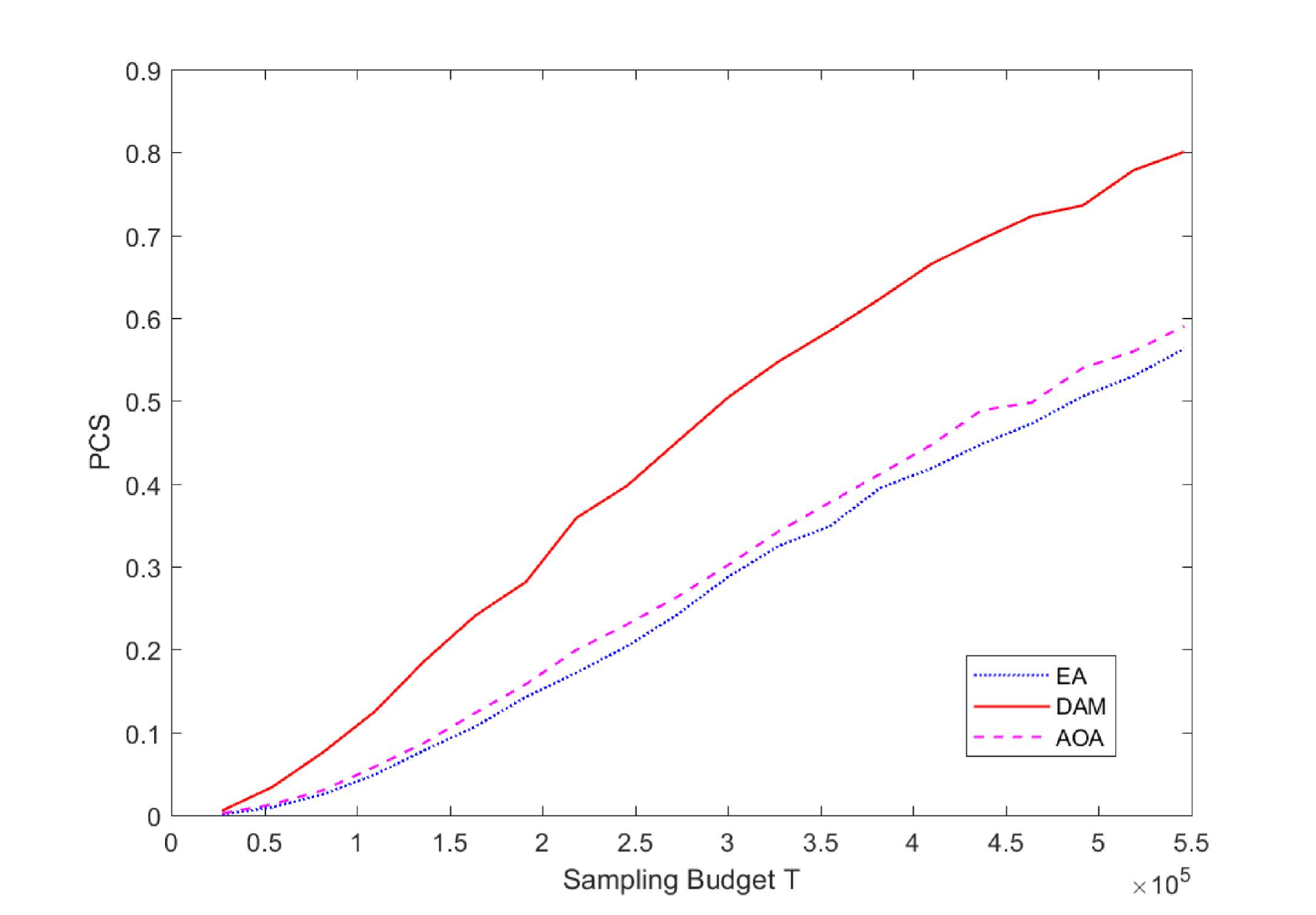}\\
		\caption{PCS of the three sampling procedures in Example 3.}\label{105-15}
	\end{center}
\end{figure}

In Figure~\ref{105-15}, we can see the PCS of the DAM grows at a much faster rate than those of the EA and AOA. In order to attain PCS = 60\%, DAM consumes less than $3.7\times10^5$ samples, while both EA and AOA require more than $5.5\times10^5$ samples. In addition, we see that the gap between the PCS of the DAM and those of the EA and AOA widens as the sampling budget increases. 

\section{Conclusions}
This paper deals with a sample allocation problem for selecting important nodes in random network. Node importance is ranked by the stationary probabilities of a Markov chain. We use the first-order Taylor expansion and normal approximation to estimate the posterior distribution of the stationary probabilities. An efficient sampling procedure named DAM is derived by maximizing a VFA one-step ahead. The sensitivity of the stationary probability with respect to each interaction parameter is taken into account in the design of DAM. Numerical experiments demonstrate that DAM is much more efficient than the other tested sampling procedures, and the performance of the proposed method is robust in different scales of the networks and the real data situation.

Unlike existing literature considering deterministic network, we focus on random network with unknown interaction parameters. Random network is a more realistic scenario of the node importance ranking problem. The proposed DAM improves the sample allocation pattern for the Markov ranking in random network, which reflects a trade-off among posterior means, variances, and sensitivities. As suggested by the numerical testing results, DAM can significantly save the sampling budget in practical applications such as Google's PageRank. 

In general, a Markov chain can have several ergodic classes and transient states, and it may not satisfy the aperiodicity condition. Decomposition for the Markov chain may be needed in order to rank the nodes in each ergodic class. Future research includes developing an efficient sampling scheme for both decomposition and ranking. Moreover, the asymptotic analysis for the sampling ratio of the sequential sampling procedure for ranking the node importance in a Markov chain also deserves future work (see \cite{peng2015asymptotic}).

%\appendices
%\section*{Appendix}

% or


\section*{Acknowledgment}
 This work was supported in part by the National Science Foundation of
China (NSFC) under Grants 71571048, 71720107003, 71690232,  and 61603321, and by the National Science
Foundation under Awards ECCS-1462409 and CMMI-1462787.

\bibliographystyle{IEEEtran}
\bibliography{ref}
%\bibliographystyle{apa}
%\bibliography{Peng}
%\bibliographystyle{IEEEtran}
%\bibliography{IEEEabrv,ref}

\ifCLASSOPTIONcaptionsoff
  \newpage
\fi

\begin{IEEEbiography}[{\includegraphics[width=1in,height=1.25in,clip,keepaspectratio]{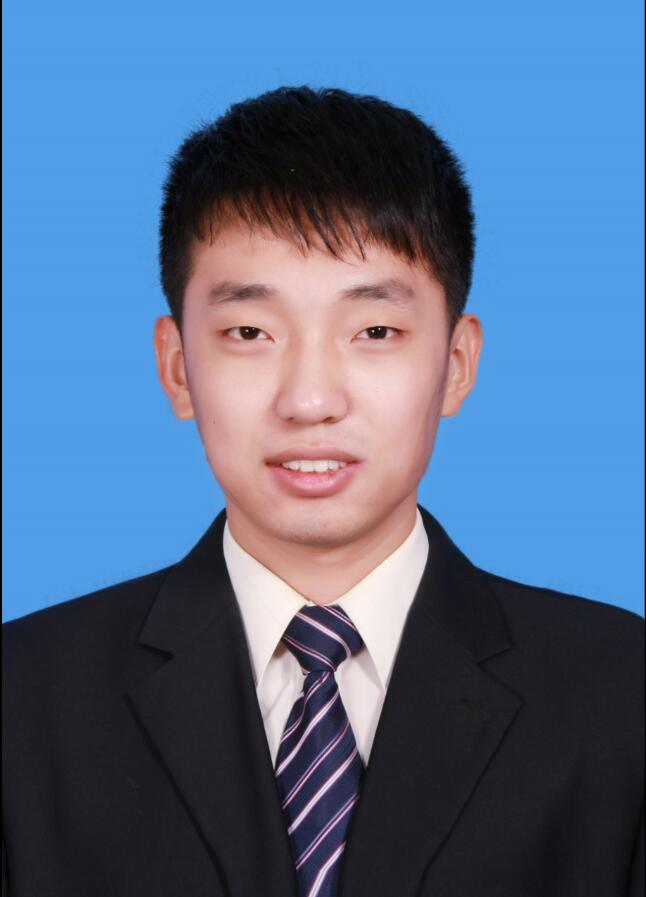}}]{Haidong Li}
is a Ph.D. candidate in the Department of Industrial Engineering and Management, Peking University, Beijing, China. He received his B.S. Degree from the Department of Engineering Mechanics at Peking University. His research interests include simulation optimization and network analysis.
\end{IEEEbiography}

\begin{IEEEbiography}[{\includegraphics[width=1in,height=1.25in,clip,keepaspectratio]{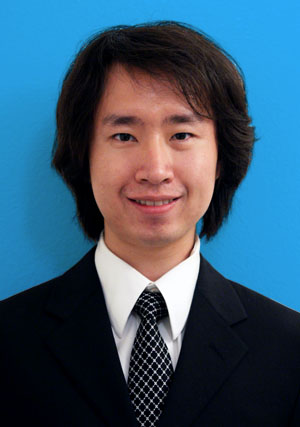}}]{Xiaoyun Xu}
received his B.S. Degree in Industrial Engineering from Tsinghua University in China in 2003, and his Ph.D. Degree in Industrial Engineering from Arizona State University in 2008. He is an Associate Professor at Department of Industrial Engineering and Management, Peking University, Beijing, China. His main research interests lie in scheduling, simulation optimization and their applications in manufacturing and service industries.
\end{IEEEbiography}

\begin{IEEEbiography}[{\includegraphics[width=1in,height=1.25in,clip,keepaspectratio]{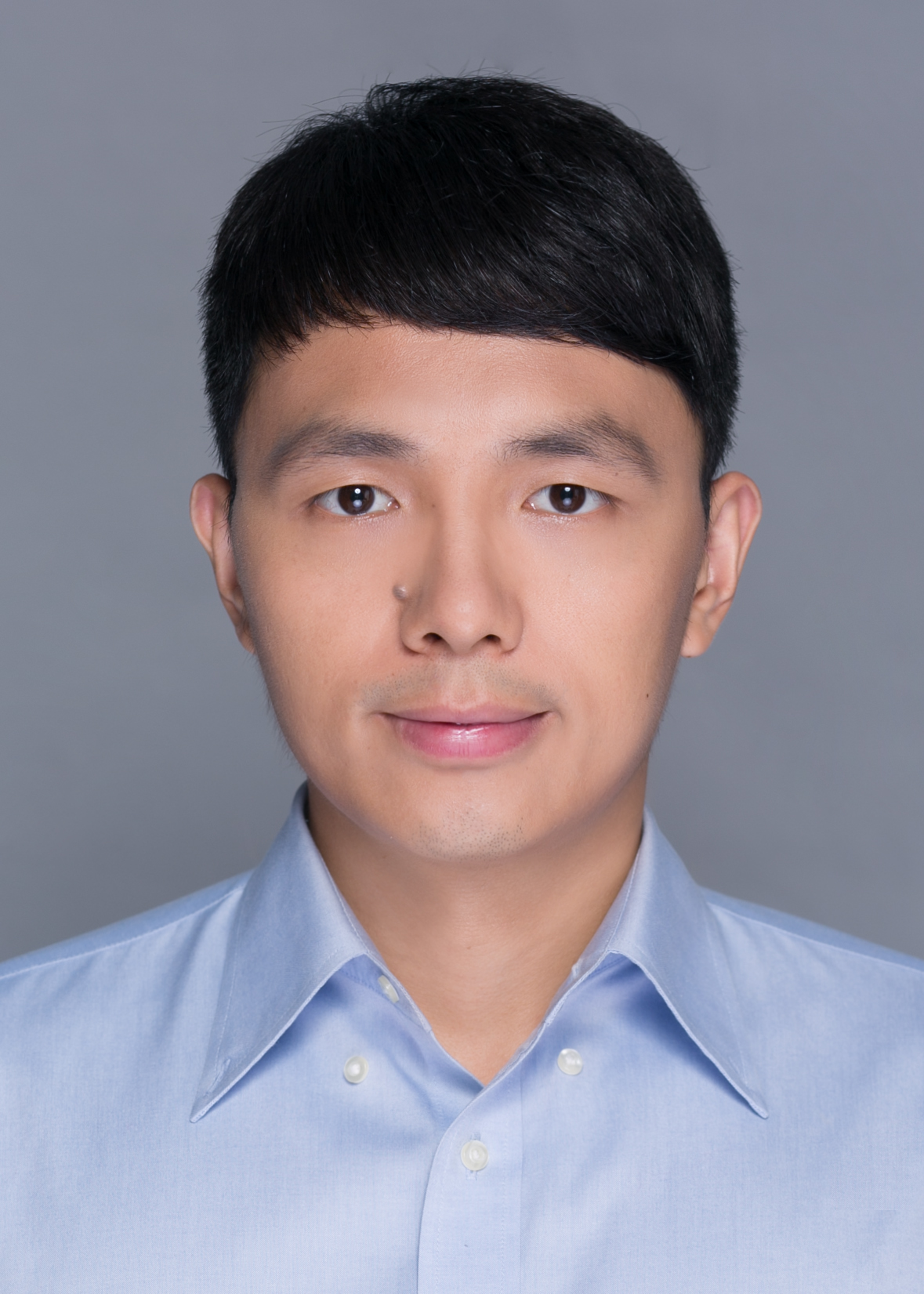}}]{Yijie Peng}
received the B.E. degree in mathematics from Wuhan University, Wuhan, China, in
2007, and the Ph.D. degree in management science from Fudan University, Shanghai, China, in 2014, respectively.He was a research fellow with Fudan University and George Mason University.

 He is currently an Assistant Professor at the Department of Industrial Engineering and Management, Peking University, Beijing, China. 
His research interests include ranking and selection and sensitivity analysis in the simulation optimization field with applications in data analytics, health care, and machine learning.
\end{IEEEbiography}

\begin{IEEEbiography}[{\includegraphics[width=1in,height=1.25in,clip,keepaspectratio]{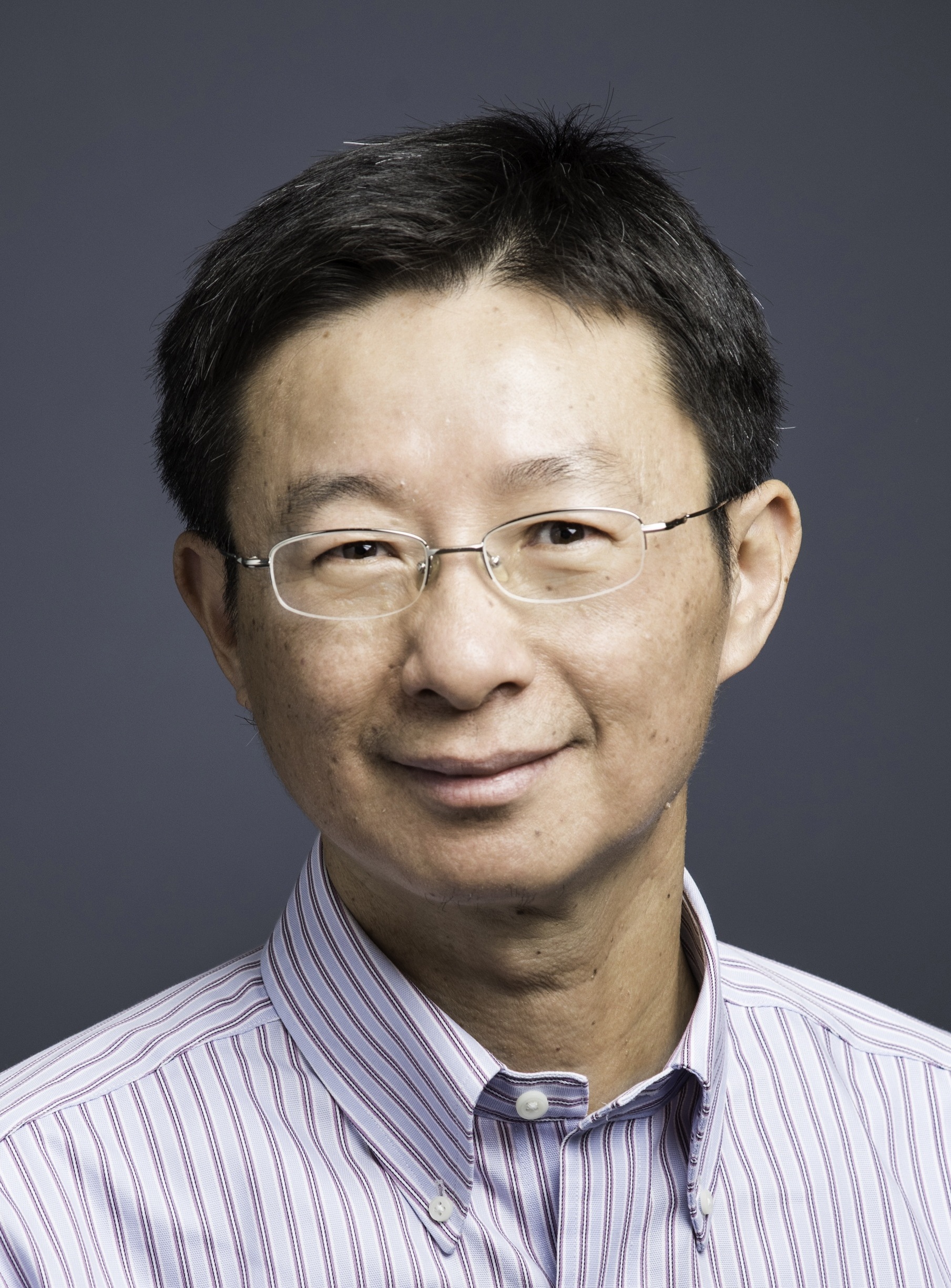}}]{Chun-Hung Chen}
received the Ph.D. degree in engineering sciences from Harvard University, Cambridge, MA, USA, in 1994. 

He is currently a Professor with the Department of Systems Engineering and Operations Research, George Mason University, Fairfax, VA, USA. He is the author of two books, including a best seller: Stochastic Simulation Optimization: An Optimal Computing Budget Allocation (World Scientific, 2010). 

Dr. Chen received the National Thousand Talents Award from the central government of China in 2011, the Best Automation Paper Award from the 2003 IEEE International Conference on Robotics and Automation, and 1994 Eliahu I. Jury Award from Harvard University. He was a Department Editor for the IIE Transactions, a Department Editor for Asia-Pacific Journal of Operational Research, an Associate Editor for the IEEE TRANSACTIONS ON AUTOMATION SCIENCE AND ENGINEERING, an Associate Editor for the IEEE TRANSACTIONS ON AUTOMATIC CONTROL, an Area Editor for the Journal of Simulation Modeling Practice and Theory, an Advisory Editor for the International Journal of Simulation and Process Modeling, and an Advisory Editor for the Journal of Traffic and Transportation Engineering.
\end{IEEEbiography}

\end{document}